\documentclass[twoside]{article}

\usepackage[accepted]{aistats2019}

\usepackage{amsthm}
\usepackage{ifamath}

\usepackage{xcolor}

\usepackage{graphicx}
 \graphicspath{{Figures/}}
 \usepackage{subcaption}
\usepackage{multicol}
\usepackage{dsfont}

\usepackage{tikz}
\usetikzlibrary{fit}
\tikzset{%
  highlight/.style={rectangle,rounded corners,fill=red!15,
    fill opacity=0.7,thick,inner sep=0pt}
}

\usepackage[customcolors]{hf-tikz}
\hfsetfillcolor{red!15}
 \hfsetbordercolor{red!30}
 \usepackage{amsmath}
\usepackage{amssymb}
\begin{document}

%

%

\twocolumn[

\aistatstitle{Bounding Inefficiency of Equilibria in Continuous Actions Games using Submodularity and Curvature }

\aistatsauthor{ Pier Giuseppe Sessa \And Maryam Kamgarpour \And  Andreas Krause }
\aistatsaddress{ ETH Zürich \And  ETH Zürich  \And ETH Zürich } ]

\begin{abstract}
Games with continuous strategy sets arise in several machine learning problems (e.g. adversarial learning). For such games, simple no-regret learning algorithms exist in several cases and ensure convergence to coarse correlated equilibria (CCE).  The efficiency of such equilibria with respect to a \emph{social function}, however, is not well understood. In this paper, we define the class of \emph{valid utility games with continuous strategies} and provide efficiency bounds for their CCEs. Our bounds rely on the social function being a monotone DR-submodular function. We further refine our bounds based on the curvature of the social function. Furthermore, we extend our efficiency bounds to a class of non-submodular functions that satisfy approximate submodularity properties. Finally, we show that valid utility games with continuous strategies can be designed to maximize monotone DR-submodular functions subject to disjoint constraints with approximation guarantees. The approximation guarantees we derive are based on the efficiency of the equilibria of such games and can improve the existing ones in the literature. We illustrate and validate our results on a budget allocation game and a sensor coverage problem. \footnote{This work was gratefully supported by Swiss National Science Foundation, under the grant SNSF $200021\_172781$, and by the European Union’s Horizon 2020 ERC grant $815943$.}
\end{abstract}
\vspace{-.5em}
\section{Introduction}

Game theory is a powerful tool for modelling many real-world multi-agent decision making problems \cite{cesa-bianchi2006}.   In machine learning, game theory has received substantial interest in the area of adversarial learning (e.g. generative adversarial networks \cite{goodfellow2014}) where models are trained via games played by competing modules~\cite{balduzzi2018}. Apart from modelling interactions among agents, game theory is also used in the context of distributed optimization. In fact, specific games can be \emph{designed} so that multiple entities can contribute to optimizing a common objective function \cite{marden2013,Li2013}.

A game is described by a set of players aiming to maximize their individual payoffs which depend on each others' strategies. The efficiency of a joint strategy profile is measured with respect to a \emph{social function}, which depends on the strategies of all the players. When the strategies for each player are uncountably infinite, the game is said to be \emph{continuous}. 

Continuous games describe a broad range of problems where integer or binary strategies may have limited expressiveness. In market sharing games \cite{goemans2006}, for instance, competing firms may invest continuous amounts in each market, or may produce an infinitely divisible product. Also, several integer problems can be generalized to continuous domains. For example, in budget allocation problems continuous amounts can be allocated to each media channel \cite{bian2017a}. In machine learning, many games are naturally continuous \cite{lee2018}. 

\vspace{-0.5em}
\subsection{Related work}
\vspace{-0.5em}

Although continuous games are finding increasing applicability, from a theoretical viewpoint they are less understood than games with finitely many strategies. Recently, no-regret learning algorithms \cite{cesa-bianchi2006} have been proposed for continuous games under different set-ups \cite{zinkevich2003,stolz2007,mertikopoulos2018}. Similarly to finite games \cite{cesa-bianchi2006}, these no-regret dynamics converge to \emph{coarse correlated equilibria} (CCEs) \cite{stolz2007,balduzzi2018}, the weakest class of equilibria which includes pure Nash equilibria, mixed Nash equilibria and correlated equilibria. However, CCEs may be highly suboptimal for the social function. A central open question is to understand the (in)efficiency of such equilibria. Differently from the finite case, where bounds on such inefficiency are known for a large variety of games \cite{roughgarden2015}, in continuous games this question is not well understood.

To measure the inefficiency of CCEs arising from no-regret dynamics, \cite{blum2008} introduces the \emph{price of total anarchy}. This notion generalizes the well-established price of anarchy (PoA) of \cite{koutsoupias1999} which instead measures the inefficiency of the worst pure Nash equilibria of the game. There are numerous reasons why players may not reach a pure Nash equilibrium \cite{blum2008,stein2011,roughgarden2015}. In contrast, regret minimization can be done by each player via simple and efficient algorithms \cite{blum2008}. 
 Recently, \cite{roughgarden2015} generalizes the price of total anarchy defining the \emph{robust} PoA which measures the inefficiency of any CCE (including the ones arising from regret minimization), and provides examples of games for which it can be bounded.


In the context of distributed optimization, where a game is designed to optimize a given objective \cite{marden2013}, bounds on the robust price of anarchy find a similar importance. In this setting, a distributed scheme to optimize the social function is to let each player implement a no-regret learning algorithm based only on its payoff information. A bound on the robust PoA provides an approximation guarantee to such optimization scheme.

Bounds on the robust PoA provided by \cite{roughgarden2015} mostly concern games with finitely many actions. A class of such games are the \emph{valid utility games} introduced by \cite{vetta2002}. In such games, the social function is a \emph{submodular} set function and, using this property, \cite{roughgarden2015} showed that the PoA bound derived in \cite{vetta2002} indeed extends to all CCEs of the game. This class of games covers numerous applications including market sharing, facility location, and routing problems, and were used by \cite{marden2013} for distributed optimization. 
Strategies consist of selecting subsets of a ground set, and can be equivalently represented as binary decisions.
Recently, authors in \cite{maehara2015} extend the notion of valid utility games to integer domains. By leveraging properties of submodular functions over integer lattices, they show that the robust PoA bound of \cite{roughgarden2015} extends to the integer case. The notion of submodularity has recently been extended to continuous domains, mainly in order to design efficient optimization algorithms \cite{bach2015,bian2017a,hassani2017}. To the best of author's knowledge, such notion has not been utilized for analyzing efficiency of equilibria of games over continuous domains.

\vspace{-.5em}
\subsection{Our contributions}
\vspace{-0.5em}

We bound the robust price of anarchy for a subclass of continuous games, which we denote as \emph{valid utility games with continuous strategies}. They are the continuous counterpart of the valid utility games introduced by \cite{vetta2002} and \cite{maehara2015} for binary and integer strategies, respectively. 
Our bounds rely on a particular game structure and on the social function being a monotone DR-submodular function \cite[Definition 1]{bian2017a}. Hence, we define the \emph{curvature} of a monotone DR-submodular function on continuous domains, analyze its properties, and use it to refine our bounds. We also show that our bounds can be extended to non-submodular functions which have `approximate' submodularity properties. This is in contrast with \cite{vetta2002,maehara2015} where only submodular social functions were considered. Finally, employing the machinery of \cite{marden2013}, we show that valid utility games with continuous strategies can be designed to maximize non convex/non concave functions in a distributed fashion with approximation guarantees. Depending on the curvature of the function, the obtained guarantees can improve the ones available in the literature. 

\vspace{-.5em}
\subsection{Notation}
\vspace{-.5em}

We denote by $\mathbf{e}_i , \mathbf{0}$, and $\mathds{1}$, the $i^{th}$ unit vector, null vector, and vector of all ones of appropriate dimensions, respectively. Given $n \in \mathbb{N}$, with $n \geq 1 $, we define $[n]:=\{1, \ldots, n\}$. Given vectors $\mathbf{x}, \mathbf{y}$, we use $[\mathbf{x}]_i$ and $x_i$ interchangeably to indicate the $i^{th}$ coordinate of $\mathbf{x}$, and $(\mathbf{x}, \mathbf{y})$ to denote the vector obtained from their concatenation, i.e., $(\mathbf{x}, \mathbf{y}):= [\mathbf{x}^\top, \mathbf{y}^\top]^\top$. Moreover, for vectors of equal dimension, $\mathbf{x} \leq \mathbf{y}$ means $x_i \leq y_i$ for all $i$. Given $\mathbf{x} \in \reals{n}$ and $j\in \{0,\ldots,n\}$, we define $[\mathbf{x}]_1^j := (x_1, \ldots, x_j, 0, \ldots, 0) \in \reals{n}$ with  $[\mathbf{x}]_1^0 = \mathbf{0}$. A function $f: \setbf{X} \subseteq \reals{n} \rightarrow \mathbb{R}$ is \emph{monotone} if, for all $\mathbf{x} \leq \mathbf{y} \in \setbf{X}$, $f(\mathbf{x}) \leq f(\mathbf{y})$. Moreover, $f$ is \emph{affine} if for all $\mathbf{x}, \mathbf{y} \in \setbf{X}$, $f(\mathbf{x}+\mathbf{y}) - f(\mathbf{x})= f(\mathbf{y})-f(\mathbf{0})$.
\vspace{-.1em}
\section{Problem formulation and examples}\label{sec:problem_formulation}
\vspace{-.5em}

We consider a class of non-cooperative continuous games, where each player $i$ chooses a vector $\mathbf{s}_i$ in its feasible strategy set $\mathcal{S}_i \subseteq \mathbb{R}_+^d$. We let $N$ be the number of players,  $\mathbf{s} = (\mathbf{s}_1, \ldots, \mathbf{s}_N)$ be the vector of all the strategy profiles, i.e., the outcome of the game, and $\setbf{S} = \prod_{i=1}^N \mathcal{S}_i  \subseteq  \mathbb{R}_+^{Nd}$ be the joint strategy space. For simplicity, we assume each strategy $\mathbf{s}_i$ is $d$-dimensional, although different dimensions could exist for different players. 
Each player aims to maximize her payoff function $\pi_i : \setbf{S} \rightarrow \mathbb{R} $, which in general depends on the strategies of all the players. We let the social function be $\gamma : \mathbb{R}_+^{Nd}   \rightarrow \mathbb{R}_+ $. For the rest of the paper we assume $\gamma(\mathbf{0}) = 0$. We denote such games with the tuple $\set{G}=(N, \{ \set{S}_i\}_{i=1}^N,  \{ \pi_i\}_{i=1}^N , \gamma)$. Given an outcome $\mathbf{s}$ we use the standard notation $(\mathbf{s}_i , \mathbf{s}_{-i})$ to denote the outcome where player $i$ chooses strategy $\mathbf{s}_i$ and the other players select strategies $\mathbf{s}_{-i} = (\mathbf{s}_1,\ldots, \mathbf{s}_{i-1}, \mathbf{s}_{i+1}, \ldots \mathbf{s}_N)$.

A pure Nash equilibrium is an outcome $\mathbf{s} \in \setbf{S}$ such that
\begingroup
\setlength{\abovedisplayskip}{0pt}
\setlength{\belowdisplayskip}{0pt}
\begin{equation*}
\pi_i(\mathbf{s})  \geq  \pi_i(\mathbf{s}_i', \mathbf{s}_{-i})  \,,
\end{equation*}
\endgroup
for every player $i$ and for every strategy $\mathbf{s}_i' \in \set{S}_i$. A coarse correlated equilibrium (CCE) is a probability distribution $\sigma$ over the outcomes $\setbf{S}$ that satisfies
\begingroup
\setlength{\abovedisplayskip}{5pt}
\setlength{\belowdisplayskip}{5pt}
\begin{equation*}
\mathbb{E}_{\mathbf{s} \sim \sigma} [ \pi_i(\mathbf{s})]  \geq \mathbb{E}_{\mathbf{s} \sim \sigma} [ \pi_i(\mathbf{s}_i', \mathbf{s}_{-i})]  \,,
\end{equation*}
\endgroup
for every player $i$ and for every strategy $\mathbf{s}_i' \in \set{S}_i$.
CCE's are the weakest class of equilibria and they include pure Nash, mixed Nash, and correlated equilibria \cite{roughgarden2015}.

Since each player selfishly maximizes her payoff, the outcome $\mathbf{s}\in \setbf{S}$ of the game is typically suboptimal for the social function $\gamma$. To measure such suboptimality, 
\cite{roughgarden2015} introduced the \emph{robust price of anarchy} (robust PoA) which measures the inefficiency of any CCE.  
Given  $\set{G}$, we let $\Delta$ be the set of all the CCEs of $\set{G}$ and define the robust PoA as the quantity
\begin{equation*}
PoA_{CCE} := \frac{\max_{\mathbf{s} \in \setbf{S}} \gamma(\mathbf{s})}{\min_{\sigma \in \Delta} \mathbb{E}_{\mathbf{s} \sim \sigma} [ \gamma(\mathbf{s})] } \,.
\end{equation*}
It can be easily seen that $PoA_{CCE} \geq 1$. As discussed in the introduction, $PoA_{CCE}$ has two important implications. In multi-agent systems, $PoA_{CCE}$ bounds the efficiency of no-regret learning dynamics followed by the selfish agents. In fact, these dynamics converge to a CCE of the game \cite{stolz2007,balduzzi2018}. In the context of distributed optimization, no-regret learning algorithms can be implemented distributively to optimize a given function and $PoA_{CCE}$ certifies the overall approximation guarantee. Bounds for $PoA_{CCE}$, however, were obtained mostly for games with finitely many actions \cite{roughgarden2015}.



In this paper, we are interested in upper bounding $PoA_{CCE}$ for continuous games $\set{G}$ defined above. To motivate our results, we present two relevant examples of such games. The first one is a budget allocation game, while in the second example a continuous game can be designed for distributed maximization in the spirit of \cite{marden2013}. We will come back to these examples in \refsec{sec:main_results} and derive upper bounds for their respective $PoA_{CCE}$'s. 

\begin{example}[\textbf{Continuous budget allocation game}]\label{example_1}
A set of $N$ advertisers enters a market consisting of a set of $d$ media channels. By allocating (or investing) part of their budget in each advertising channel, the goal of each advertiser is to maximize the expected number of \emph{activated} customers, i.e., customers who purchase her product. The market is described by a bipartite graph $G = (\set{R} \cup \set{T}, \set{E})$, where the left vertices $\set{R}$ denote channels and the right vertices $\set{T}$ denote customers, with $d = \vert \set{R} \vert$. For each advertiser $i$ and edge $(r,t) \in \set{E}$, $p_i(r,t)\in [0,1]$ is the probability that advertiser $i$ activates customer $t$ via channel $r$. Each advertiser chooses a strategy $\mathbf{s}_i \in \reals{d}_+$, which represents the amounts allocated (or invested) to each channel, subject to budget constraints $\set{S}_i = \{ \mathbf{s}_i \in \reals{d} : \mathbf{c}_i^\top \mathbf{s}_i \leq b_i , \mathbf{0} \leq \mathbf{s}_i \leq \bar{\mathbf{s}}_i\}$. This generalizes the set-up in \cite{maehara2015}, where strategies $\mathbf{s}_i$ are integer. Hence, we consider the continuous version of the game modeled by \cite{maehara2015}.
For every customer $t\in \set{T}$ and advertiser $i \in [N]$, we define $\Gamma(t) = \{ r \in \set{R} : (r,t)\in \set{E} \}$ and the quantity 
\begingroup
\setlength{\abovedisplayskip}{5pt}
\setlength{\belowdisplayskip}{5pt}
\begin{equation*}
P_i(\mathbf{s}_i, t) = 1 - \prod\nolimits_{r \in \Gamma(t)} (1-p_i(r,t))^{[\mathbf{s}_i]_r} \, ,
\end{equation*}
\endgroup
which is the probability that $i$ activates $t$ when the other advertisers are ignored. For each customer $t$, a permutation $\rho \in \set{P}_N$ is drawn uniformly at random, where $\set{P}_N$ is the set of all permutations of $[N]$. Then, according to $\rho$ each advertiser sequentially attempts to activate customer $t$. Hence, for a given allocation $\mathbf{s} = (\mathbf{s}_1, \ldots, \mathbf{s}_N) \in \setbf{S}=\prod_{i=1}^N \set{S}_i$, the payoff of each advertiser can be written in closed form as \cite{maehara2015}:
\begingroup
\setlength{\abovedisplayskip}{5pt}
\setlength{\belowdisplayskip}{5pt}
\begin{equation*}
\pi_i(\mathbf{s}) = \frac{1}{N !} \sum\nolimits_{t \in \set{T}} \sum\nolimits_{\rho \in \set{P}_N} P_i(\mathbf{s}_i,t) \prod_{j \prec_\rho i}(1- P_j(\mathbf{s}_j,t)) \, ,
\end{equation*}
\endgroup
where $j \prec_\rho i$ indicates that $j$ precedes $i$ in $\rho$. The term $\pi_i(\mathbf{s})$ represents the expected number of customers activated by advertiser $i$ in allocation $\mathbf{s}$. 
The goal of the market analyst, which assumes the role of the game planner, is to maximize the expected number of customers activated. Hence, for any  $\mathbf{s}$, the social function $\gamma$ is
\begingroup
\setlength{\abovedisplayskip}{5pt}
\setlength{\belowdisplayskip}{5pt}
\begin{equation*}
\gamma(\mathbf{s}) = \sum\nolimits_{i=1}^N \pi_i(\mathbf{s}) = \sum\nolimits_{t \in \set{T}} \Big(  1 - \prod\nolimits_{i=1}^N(1- P_i(\mathbf{s}_i,t)) \Big) \,.
\end{equation*} 
\endgroup
\end{example}

\begin{example}[\textbf{Sensor coverage with continous assignments}]\label{example_2}
Given a set of $N$ autonomous sensors, we seek to monitor a finite set of $d$ locations in order to maximize the 
probability of detecting an event. For each sensor, a continuous variable $x_i \in \reals{d}_+$ indicates the energy assigned (or time spent) to each location, subject to budget constraints $\set{X}_i := \{ \mathbf{x}_i \in \reals{d}: \mathbf{c}_i^\top \mathbf{x}_i \leq b_i , \mathbf{0} \leq \mathbf{x}_i \leq \bar{\mathbf{x}}_i\}$. This generalizes the well-known sensor coverage problem studied in \cite{marden2013} (and previous works), where $\mathbf{x}_i$'s are binary and indicate the locations sensor $i$ is assigned to. The probability that sensor $i$ detects an event in location $r$ is $1- (1-p_i^r)^{[x_i]_r}$, with $0 \leq p_i^r \leq 1$, and it increases as more energy is assigned to the location. Hence, given a strategy $\mathbf{x} = (\mathbf{x}_1, \ldots, \mathbf{x}_N)$, the joint probability of detecting an event in location $r$ is
\begingroup
\setlength{\abovedisplayskip}{5pt}
\setlength{\belowdisplayskip}{5pt}
\begin{equation*}
P(r,\mathbf{x}) = 1 - \prod\nolimits_{i \in [N]}(1-p_i^r)^{[\mathbf{x}_i]_r} \,.
\end{equation*}
\endgroup
The goal of the planner is to maximize the probability of detecting an event
\begingroup
\setlength{\abovedisplayskip}{5pt}
\setlength{\belowdisplayskip}{5pt}
\begin{equation*}
\gamma(\mathbf{x}) = \sum\nolimits_{r \in [d]}w_r \: P(r,\mathbf{x})\,,
\end{equation*}
\endgroup
where $w_r$'s represent the a priori probability that an event occurs in location $r$. As in \cite{marden2013}, we can set up a continuous game $\set{G}=(N, \{ \set{S}_i\}_{i=1}^N,  \{ \pi_i\}_{i=1}^N , \gamma)$ where $\set{S}_i = \set{X}_i$ for each $i$, and $\pi_i$'s are designed so that good monitoring solutions can be obtained when each player selfishly maximizes her payoff. 
\end{example}

\emph{Proofs of the upcoming propositions and remarks are presented in \refapp{A}}.
\vspace{-0.5em}
\section{Main results}\label{sec:main_results}
\vspace{-0.5em}

We derive $PoA_{CCE}$ bounds for a sublass of continuous games $\set{G}$ by extending the valid utility games considered in \cite{vetta2002} and \cite{maehara2015} to continuous strategy sets.
Hence, in \refdef{def:valid_continuous_game} we will define the class of \emph{valid utility games with continuous strategies}. As will be seen, the two problems described above and several other examples fall into this class. 
At the end of the section, we will show that valid utility games can be designed to maximize non-convex/non-concave objectives in a distributed fashion with approximation guarantees.  
\vspace{-1em}
\subsection{Robust PoA bounds}
\vspace{-0.5em}

As in \cite{vetta2002,maehara2015}, the $PoA_{CCE}$ bounds obtained rely on the social function $\gamma$ experiencing diminishing returns (DR). Differently from set functions, in continuous (and integer) domains, different notions of DR exist. Similarly to \cite{maehara2015}, our first main result relies on $\gamma$ satisfying the strongest notion of DR, also known as \texttt{DR} property \cite{bian2017a}, which we define in \refdef{def:DR_property}. 
Moreover, as in \cite{vetta2002} our bound can be refined depending on the  \emph{curvature} of $\gamma$. While DR properties have been recently studied also in continuous domains, notions of curvature of a submodular function were only explored for set functions \cite{conforti1984,iyer2013} (see \cite[Appendix C]{bian2017b} for a comparison of the existing notions). Hence, in \refdef{def:curvature} we define the curvature of a monotone DR-submodular function on continuous domains.

\begingroup
\setlength{\abovedisplayskip}{6pt}
\setlength{\belowdisplayskip}{6pt}
\begin{definition}[\texttt{DR} property]\label{def:DR_property}
  A function $f: \setbf{X} = \prod_{i=1}^n \set{X}_i \rightarrow \mathbb{R}$ with $\set{X}_i \subseteq \reals{}$ is \emph{DR-submodular} if for all $\mathbf{x} \leq \mathbf{y} \in \setbf{X}$, $\forall i \in [n], \forall k \in \mathbb{R}_+$ such that $(\mathbf{x} + k \mathbf{e}_i)$ and $(\mathbf{y} + k \mathbf{e}_i)$ are in $ \setbf{X}$,
  \vspace{-.5em}
\begin{equation*}
f(\mathbf{x} + k \mathbf{e}_i ) - f(\mathbf{x})  \geq    f(\mathbf{y} + k \mathbf{e}_i) - f(\mathbf{y}) \,.
\end{equation*} 
\end{definition} 
When restricted to binary sets $\setbf{Z} = \{0,1\}^n$, \refdef{def:DR_property} coincides with the standard notion of submodularity for set functions.
An equivalent characterization of the \texttt{DR} property for a twice-differentiable function is that all the entries of its Hessian are non-positive \cite{bian2017a}:
\begin{equation*}
\forall \mathbf{x} \in \setbf{X}, \quad \frac{\partial^2 f(\mathbf{x})}{\partial x_i \partial x_j} \leq 0 ,\quad \forall i, j  \, .
\end{equation*}
\endgroup

 \begingroup
\setlength{\abovedisplayskip}{5pt}
\setlength{\belowdisplayskip}{5pt}
\begin{definition}[curvature]\label{def:curvature}Given a monotone DR-submodular function $f: \setbf{X} \subseteq \reals{n}_+  \rightarrow \mathbb{R}$, and a set $\setbf{Z} \subseteq \setbf{X}$ with $\mathbf{0} \in \setbf{Z}$, we define the curvature of $f$ with respect to $\setbf{Z}$ by
\vspace{-.5em}
\begin{equation*}
  \alpha( \setbf{Z}) = 1 - \inf_{\substack{\mathbf{x}\in \setbf{Z} , i \in [n]: \\\mathbf{x}+ k \mathbf{e}_i \in \setbf{Z}  }} \: {\lim_{k\rightarrow 0^+}\frac{f(\mathbf{x}+ k \mathbf{e}_i) - f(\mathbf{x})}{f(k \mathbf{e}_i)- f(\mathbf{0})}}   \,.
 \end{equation*}
\end{definition} 
\begin{remark} \label{remark:curvature_in_zero_one}
For any monotone function  $f: \reals{n}  \rightarrow \mathbb{R}$ and $\forall \setbf{Z} \subseteq \reals{n}$ with $\mathbf{0} \in \setbf{Z}$, $\alpha(\setbf{Z}) \in [0,1]$.
\end{remark}
When restricted to binary sets $\setbf{Z} = \{0,1\}^n$, \refdef{def:curvature} coincides with the \emph{total curvature} defined in \cite{conforti1984}. Moreover, if $f$ is montone DR-submodular and \emph{differentiable}, its curvature with respect to a set $\setbf{Z}$ can be computed as: 
\begin{equation*}
 \alpha( \setbf{Z}) = 1 - \inf\nolimits_{\substack{\mathbf{x}\in \setbf{Z} \\ i \in [n]}}\frac{[\nabla f(\mathbf{x})]_i}{[\nabla f(\mathbf{0})]_i} \,.
\end{equation*}
\endgroup


Based on the previous definitions, we define the class of valid utility games with continuous strategies.

\begin{definition}\label{def:valid_continuous_game}
A game $\set{G}=(N, \{ \set{S}_i\}_{i=1}^N,  \{ \pi_i\}_{i=1}^N , \gamma)$ is a \emph{valid utility game with continuous strategies} if: 
\vspace{-.5em}
\begin{itemize}
\item[i)] The function $\gamma$ is monotone DR-submodular.
\vspace{-.5em}

\item[ii)]For each player $i$ and for every outcome $\mathbf{s}$,  $\pi_i(\mathbf{s} ) \geq \gamma(\mathbf{s}) - \gamma(\mathbf{0}, \mathbf{s}_{-i}) $. 
\vspace{-.5em}

\item[iii)]For every outcome $\mathbf{s}$, $\gamma(\mathbf{s})  \geq \sum_{i=1}^N \pi_i(\mathbf{s})$.
\end{itemize}
\end{definition}
\vspace{-.5em}

Intuitively, the conditions above ensure that the payoff for each player is at least her contribution to the social function and that optimizing $\gamma$ is somehow bind to the goals of the players. Defining the set $\setbf{\tilde{S}} :=\{ \mathbf{x} \in \reals{Nd}_+ \mid \mathbf{0} \leq \mathbf{x} \leq  \mathbf{s}_{max} \}$, with $\mathbf{s}_{max}$ such that $ \forall \mathbf{s},\mathbf{s}' \in \setbf{S} , \, \mathbf{s} +  \mathbf{s}' \leq \mathbf{s}_{max} $, we can establish the following main theorem. 

\begin{theorem}\label{thm:main_theorem}
 Let $\set{G}=(N, \{ \set{S}_i\}_{i=1}^N,  \{ \pi_i\}_{i=1}^N , \gamma)$ be a valid utility game with continuous strategies with social function $\gamma :\reals{Nd}_+ \rightarrow \reals{}_+
$ having \emph{curvature} $\alpha(\setbf{\tilde{S}}) \leq \alpha$. Then, $PoA_{CCE} \leq (1 + \alpha)$. 
\end{theorem}
\vspace{-0.5em}

We will prove \refthm{thm:main_theorem} in \refsec{sec:proof_of_thm1}. 
\begin{remark}\label{remark:bound_2}
If $\set{G}$ is a valid utility game with continuous strategies, then $PoA_{CCE} \leq 2$.
\end{remark}


\begin{remark}
The notion of valid utility games above is an exact generalization of the one by \cite{maehara2015} for integer strategy sets. Leveraging recent advances in `approximate' submodular functions, in \refsec{sec:extension} we relax condition i) and derive $PoA_{CCE}$ bounds for a strictly larger class of games.
\end{remark}

Using \refthm{thm:main_theorem}, the following proposition upper bounds $PoA_{CCE}$ of \refexa{example_1}. Our bound depends on the activation probabilities $p_i(r,t)$'s and on the connectivity of the market $G$.

\begin{proposition}\label{prop:Example_1} The budget allocation game defined in \refexa{example_1} is a valid utility game with continuous strategies. Moreover, $PoA_{CCE} \leq 1 + \alpha <2$ with 
\begingroup
\setlength{\abovedisplayskip}{0pt}
\setlength{\belowdisplayskip}{0pt}
\begin{align*}
& \alpha := 1 - \\& \min_{ i \in [N], r\in [d]} \frac{  \sum\limits_{t \in \set{T}: r \in \Gamma(t)} \ln(1-p_i(r,t))\prod\limits_{j \in [N]}(1-P_j(2\bar{\mathbf{s}}_j,t))}{ \sum\limits_{t \in \set{T}: r \in \Gamma(t)} \ln(1-p_i(r,t)) } \,.
\end{align*}
\end{proposition}
In our more general continuous actions framework, the obtained bound strictly improves the bound of 2 by \cite{maehara2015}, since the curvature of the social function was not considered in \cite{maehara2015}. We will visualize our bound in the experiments of \refsec{sec:examples}.
\endgroup

Using \refthm{thm:main_theorem}, we now generalize \refexa{example_2} and show that valid utility games with continuous strategies can be designed to maximize monotone DR-submodular functions subject to decoupled constraints with approximation guarantees. The proposed optimization scheme will be used in \refsec{sec:examples} to solve an instance of the sensor coverage problem defined in \refexa{example_2}.

\vspace{-0.5em}
\subsection{Game-based monotone DR-submodular maximization}\label{sec:distributed_maximization}
\vspace{-0.5em}

Consider the general problem of maximizing a monotone DR-submodular function $\gamma: \reals{n} \rightarrow \reals{}_+$ subject to \emph{decoupled} constraints $\setbf{X} = \prod_{i=1}^N \set{X}_i \subseteq \reals{n}$. We can assume $\set{X}_i \subseteq \reals{}_+$ without loss of generality \cite{bian2017a}, since otherwise one could optimize $\gamma$ over a shifted version of its constraints. Moreover, we assume $\gamma(\mathbf{0}) = 0$ for ease of exposition. Note that the class of monotone DR-submodular functions includes non concave functions.  To find approximate solutions, we set up a game 
\begingroup
\setlength{\abovedisplayskip}{9pt}
\setlength{\belowdisplayskip}{9pt}
\begin{equation*}
\hat{\set{G}} :=(N, \{ \hat{\set{S}}_i\}_{i=1}^N,  \{ \hat{\pi}_i\}_{i=1}^N , \gamma) \, ,
\end{equation*} 
\endgroup
where for each player $i$, $\hat{\set{S}}_i := \set{X}_i$, and $\hat{\pi}_i(\mathbf{s}) := \gamma(\mathbf{s}) - \gamma(\mathbf{0}, \mathbf{s}_{-i})$ for every outcome $\mathbf{s} \in \setbf{S} = \setbf{X}$. By using DR-submodularity of $\gamma$, we can affirm the following. 

\textbf{Fact 1.}
$\hat{\set{G}}$ is a valid utility game with continuous strategies. 
\vspace{-0.5em}

Assume there exists $\mathbf{x}_{max} \in \reals{n}_+$ such that $ \forall \mathbf{x},\mathbf{x}' \in \setbf{X} , \, \mathbf{x} +  \mathbf{x}' \leq \mathbf{x}_{max}$. Then, we denote with $\alpha(\setbf{\tilde{X}})$ the curvature of $\gamma$ with respect to $\setbf{\tilde{X}} :=\{ \mathbf{x} \in \reals{n}_+ \mid \mathbf{0} \leq \mathbf{x} \leq  \mathbf{x}_{max} \}$ and let $\alpha \in [0,1]$ be an upper bound for $\alpha(\setbf{\tilde{X}})$. If such $\mathbf{x}_{max}$ does not exist, we let $\alpha = 1$.
Moreover, assume that for each player $i \in [N]$ there exists a no-regret algorithm \cite[Sec. 3]{monnot2017} to play $\hat{\set{G}}$. That is, when $\hat{\set{G}}$ is repeated over time, player $i$ can ensure that $\frac{1}{T}\max_{s \in \hat{\set{S}}_i} \sum_{t=1}^T \hat{\pi_i}(s, \mathbf{s}_{-i}^t) - \frac{1}{T} \sum_{t=1}^T \hat{\pi_i}(s_i^t, \mathbf{s}_{-i}^t) \rightarrow 0$ as $T \rightarrow \infty$, for any sequence $\{ \mathbf{s}_{-i}^t \}_{t=1}^T$.
We let \textsc{D-noRegret} be the distributed algorithm where such no-regret algorithms are simultaneously implemented for each player. 
We can establish the following corollary of \refthm{thm:main_theorem} \footnote{A similar version of the corollary can be obtained when no-$\alpha$-regret \cite[Definition 4]{kakade2007} algorithms exist for each player, such as the ones by \cite{chen2018a,chen2018b} for online submodular maximization.}

\begin{corollary}\label{cor:D_no_regret}
Let $\mathbf{x}^\star = \argmax_{\mathbf{x}\in \setbf{X}}\gamma(\mathbf{x})$. Then,\\\textsc{D-noRegret} converges to a distribution $\sigma$ over $\setbf{X}$ such that $\mathbb{E}_{\mathbf{x} \sim \sigma} [ \gamma(\mathbf{x})] \geq 1/( 1 + \alpha)\gamma(\mathbf{x}^\star)$.
\end{corollary}
\vspace{-0.5em}

Note that the \textsc{Frank-Wolfe} variant of \cite{bian2017a} can also be used to maximize $\gamma$ with $(1-e^{-1})$ approximations, under the additional assumption that $\setbf{X}$ is down-closed. For small $\alpha$'s, however, our guarantee can strictly improve the one by \cite{bian2017a}. 

If $\setbf{X}$ is convex compact and $\gamma$ is concave in each $\set{X}_i$, then $\hat{\pi}_i$'s are concave in each $x_i$ and the online gradient ascent algorithm by \cite{zinkevich2003} ensures no-regret for each player \cite{flaxman2005}. Using \refcor{cor:D_no_regret}, we show that the sensor coverage problem of \refexa{example_2} falls into this class and \textsc{D-noRegret} has approximation guarantees that depend on the sensing probabilities $P(r, \cdot)$'s.  
\begingroup
\setlength{\abovedisplayskip}{3pt}
\setlength{\belowdisplayskip}{3pt}
\begin{proposition}\label{prop:Example_2}
Consider the sensor coverage problem of \refexa{example_2} and assume we set-up the game $\hat{\set{G}}$. Then, online gradient ascent \cite{zinkevich2003} is a no-regret algorithm for each player. Moreover, \textsc{D-noRegret} has an expected approximation ratio of $1/(1+\alpha)$, where $\alpha := \max_{r \in [d]}P(r, 2\bar{\mathbf{x}})$ and $\bar{\mathbf{x}} = (\bar{\mathbf{x}}_1,\ldots, \bar{\mathbf{x}}_N)$.
\end{proposition}
\vspace{-0.5em}

Note that the obtained approximation ratio is strictly larger than $\frac{1}{2}$ and it increases when the number of sensors $N$ or the detection probabilities decrease, a fact also noted in \cite{marden2013} for the binary setting. We compare the performance of \textsc{D-noRegret} and the \textsc{Frank-Wolfe} variant of \cite{bian2017a} in \refsec{sec:examples}.
\endgroup

A decentralized maximization scheme for submodular functions is also proposed in \cite{mokhtari2018}, albeit in a different setting. In \cite{mokhtari2018}, $\gamma$ consists of a sum of \emph{local} functions subject to a common down-closed convex contraint set, while we considered a generic objective $\gamma$ subject to local constraints.



\vspace{-.5em}
\section{Analysis} \label{sec:background}
\vspace{-.5em}

In order to prove \refthm{thm:main_theorem} and its extension to non-submodular functions (\refsec{sec:extension}), we first review the main properties of submodularity in continuous domains. We will introduce two straightforward interpretations of the DR properties defined in the literature, and show a fundamental property of the curvature of a monotone DR-submodular function. 

\vspace{-0.5em}
\subsection{Submodularity and curvature on continuous domains}
\vspace{-0.5em}
Submodularity in continuous domains has received recent attention for approximate maximization and minimization of non convex/non concave functions \cite{bian2017a,hassani2017,bach2015}.
Submodular continuous functions are defined on subsets of $\reals{n}$ of the form $\setbf{X} = \prod_{i=1}^n \set{X}_i$, where each $\set{X}_i$ is a compact subset of $\reals{}$. 
A function $f: \setbf{X} \rightarrow \reals{}$ is \emph{submodular} if
for all $\mathbf{x} \in \setbf{X}$, $\forall i, j$ and $a_i,a_j >0$ s.t. $x_i + a_i \in \mathcal{X}_i$, $x_j + a_j \in \mathcal{X}_j$, \cite{bach2015}
\begingroup
\setlength{\abovedisplayskip}{5pt}
\setlength{\belowdisplayskip}{5pt}
\begin{equation*}
f(\mathbf{x} + a_i \mathbf{e}_i ) - f(\mathbf{x})  \geq    f(\mathbf{x} + a_i \mathbf{e}_i+  a_j \mathbf{e}_j) - f(\mathbf{x} + a_j \mathbf{e}_j) \,.
\end{equation*}
\endgroup
The above property also includes submodularity of set functions, by restricting $\set{X}_i$'s to $\{0,1\}$, and over integer lattices, by restricting $\set{X}_i$'s to $\mathbb{Z}_+$. We are interested, however, in submodular continuous functions, where $\set{X}_i$'s are compact subsets of $\reals{}$.
As thoroughly studied for set functions, submodularity is related to diminishing return properties of $f$. However, differences exist when considering functions over continuous (or integer) domains. 
In particular, submodularity is equivalent to the following \emph{weak} DR property \cite{bian2017a}.

\begin{definition}[\texttt{weak DR} property]\label{def:weak_DR_property}
  A function $f: \setbf{X}\subseteq \reals{n}  \rightarrow \mathbb{R}$ is \emph{weakly DR-submodular} if, for all $\mathbf{x} \leq \mathbf{y} \in \setbf{X}$, $\forall i$ s.t. $x_i = y_i$, $\forall k \in \mathbb{R}_+$ s.t. $(\mathbf{x} + k \mathbf{e}_i)$ and $(\mathbf{y} + k \mathbf{e}_i)$ are in $ \setbf{X}$, 
\begingroup
\setlength{\abovedisplayskip}{5pt}
\setlength{\belowdisplayskip}{5pt}\begin{equation*}
f(\mathbf{x} + k \mathbf{e}_i ) - f(\mathbf{x})  \geq    f(\mathbf{y} + k \mathbf{e}_i) - f(\mathbf{y}) \,.
\end{equation*}
\end{definition} 
\endgroup
The \texttt{DR} property, which we defined in \refdef{def:DR_property} characterizes the full notion of diminishing returns and indentifies a subclass of submodular continuous functions. While \texttt{weak DR} and \texttt{DR} properties coincide for set functions, this is not the case for functions on integer or continuous lattices. As the next section reveals, the \texttt{weak DR} property of $\gamma$ is indeed not sufficient to prove \refthm{thm:main_theorem}. However, it will be useful in \refsec{sec:extension} when we extend our results to non-submodular functions.  In \refapp{app:properties_submodular} we discuss submodularity for differentiable functions.

To prove the main results of the paper, the following two propositions provide equivalent characterizations of \texttt{weak DR} and \texttt{DR} properties, respectively\footnote{The introduced properties are the continuous versions of the `group DR property'\cite{bilmes_slides2011} of submodular set functions.}.

\begingroup
\setlength{\abovedisplayskip}{5pt}
\setlength{\belowdisplayskip}{5pt}
\begin{proposition}\label{prop:group_weak_DR_property}
  A function $f: \setbf{X}\subseteq \reals{n}  \rightarrow \mathbb{R}$ is \emph{weakly DR-submodular} (\refdef{def:weak_DR_property}) if and only if for all $\mathbf{x} \leq \mathbf{y} \in \setbf{X}$, $\forall \mathbf{z} \in \reals{n}_+$ s.t. $(\mathbf{x} + \mathbf{z})$ and $(\mathbf{y} + \mathbf{z})$ are in $\setbf{X}$, with $z_i = 0$ $\forall i \in [n] : y_i > x_i $, 
  \vspace{-.5em}
\begin{equation*}
f(\mathbf{x} + \mathbf{z} ) - f(\mathbf{x})  \geq    f(\mathbf{y} +\mathbf{z}) - f(\mathbf{y}) \,.
\end{equation*}
\end{proposition} 
  \vspace{-.5em}

\begin{proposition}\label{prop:group_DR_property}
  A function $f: \setbf{X} \subseteq \reals{n} \rightarrow \mathbb{R}$ is \emph{DR-submodular} (\refdef{def:DR_property}) if and only if for all $\mathbf{x} \leq \mathbf{y} \in \setbf{X}$, $\forall \mathbf{z} \in \reals{n}_+$ s.t. $(\mathbf{x} + \mathbf{z})$ and $(\mathbf{y} + \mathbf{z})$ are in $\setbf{X}$, 
  \vspace{-.5em}
\begin{equation*}
f(\mathbf{x} + \mathbf{z} ) - f(\mathbf{x})  \geq    f(\mathbf{y} +\mathbf{z}) - f(\mathbf{y}) \,.
\end{equation*}
\end{proposition} 
\endgroup

Finally, the following proposition clarifies the role of the curvature of a DR-submodular function and is key for the proof of \refthm{thm:main_theorem}.

\begingroup
\setlength{\abovedisplayskip}{5pt}
\setlength{\belowdisplayskip}{5pt}
\begin{proposition}\label{prop:equivalence_curvature}
Consider a monotone DR-submodular function $f: \setbf{X} \subseteq \reals{n}_+  \rightarrow \mathbb{R}$, and a set $\setbf{Z}:=\{\mathbf{x} \in \reals{n} : \mathbf{0} \leq \mathbf{x} \leq \mathbf{z}_{max} \} \subseteq \setbf{X}$. Then, for any $\mathbf{x}, \mathbf{y} \in \setbf{Z}$ such that $\mathbf{x} + \mathbf{y} \in \setbf{Z}$, 
 \begin{equation*}
 f(\mathbf{x} + \mathbf{y}) - f(\mathbf{x}) \geq (1-\alpha( \setbf{Z}))[f(\mathbf{y}) - f(\mathbf{0})] \, ,
 \end{equation*}
 where $\alpha( \setbf{Z})$ is the curvature of $f$ with respect to $\setbf{Z}$.
\end{proposition}
\endgroup
%
%
%
\vspace{-0.5em}
\subsection{Proof of Theorem 1}\label{sec:proof_of_thm1}
\vspace{-0.5em}


The proof uses submodularity of the social function similarly to \cite[Example 2.6]{roughgarden2015} and  \cite[Proposition 4]{maehara2015}. However, it allows us to consider the curvature of $\gamma$. Differently from  \cite[Proposition 4]{maehara2015}, our proof does not rely on the structure of the strategy sets $\mathcal{S}_i$'s. The \texttt{weak DR} and \texttt{DR} properties are used separately in the proof, to show that the \texttt{weak DR} property of $\gamma$ is not sufficient to obtain the results. This fact was similarly noted in \cite{maehara2015} for the integer case. 

To upper bound $PoA_{CCE}$, we first prove that for any pair of outcomes $\mathbf{s}, \mathbf{s}^\star \in \setbf{S}$,
\begingroup
\setlength{\abovedisplayskip}{6pt}
\setlength{\belowdisplayskip}{6pt}
\begin{equation*}
\sum\nolimits_{i=1}^N \pi_i(\mathbf{s}_i^\star, \mathbf{s}_{-i}) \geq \gamma( \mathbf{s}^\star ) -  \alpha \: \gamma(\mathbf{s}) \,.
\end{equation*}
\endgroup
In the framework of \cite{roughgarden2015}, this means that $\set{G}$ is a \emph{($ 1,  \alpha$)-smooth} game. Then, few inequalities from \cite{roughgarden2015} show that $PoA_{CCE} \leq (1 + \alpha)$. 

The smoothness proof is obtained as follows.
Consider any pair of outcomes $\mathbf{s},\mathbf{s}^\star \in \setbf{S}$. For $i \in \{0,\ldots, N\}$ with a slight abuse of notation we define $[ \mathbf{s}^\star]_1^i = (\mathbf{s}^\star_1, \ldots, \mathbf{s}^\star_i, \mathbf{0}, \ldots, \mathbf{0})$ with $[\mathbf{s}^\star ]_1^0 = \mathbf{0}$, where $\mathbf{s}^\star_j$ is the strategy of player $j$ in the outcome $\mathbf{s}^\star$. We have: 
\begingroup
\setlength{\abovedisplayskip}{4pt}
\setlength{\belowdisplayskip}{4pt}
\begin{align*}
&\sum\nolimits_{i=1}^N \pi_i(\mathbf{s}_i^\star, \mathbf{s}_{-i}) \geq \sum\nolimits_{i=1}^N \gamma(\mathbf{s}_i^\star, \mathbf{s}_{-i}) - \gamma(\mathbf{0}, \mathbf{s}_{-i})  \\ 
& \geq   \sum\nolimits_{i=1}^N \gamma(\mathbf{s}_i^\star + \mathbf{s}_i, \mathbf{s}_{-i}) - \gamma(\mathbf{s})  \\
& \geq   \sum\nolimits_{i=1}^N \gamma( \mathbf{s} + [ \mathbf{s}^\star ]_1^i) - \gamma(\mathbf{s} + [ \mathbf{s}^\star]_1^{i-1})  \\
& =  \gamma(\mathbf{s} + \mathbf{s}^\star ) - \gamma(\mathbf{s}) \\
& \geq  (1-\alpha)\gamma(\mathbf{s}) + \gamma( \mathbf{s}^\star ) - \gamma(\mathbf{s})   =   \gamma( \mathbf{s}^\star ) -  \alpha \: \gamma(\mathbf{s}) \,. 
\end{align*}
\endgroup
The first inequality follows from condition ii) of valid utility games as per \refdef{def:valid_continuous_game}. The second inequality from $\gamma$ being DR-submodular (and using \refprop{prop:group_DR_property}). The third inequality from $\gamma$ being weakly DR-submodular (and using \refprop{prop:group_weak_DR_property}).
The last inequality follows since, by \refprop{prop:equivalence_curvature},
\begingroup
\setlength{\abovedisplayskip}{6pt}
\setlength{\belowdisplayskip}{6pt}\begin{equation*}
\gamma(\mathbf{s} + \mathbf{s}^\star)  - \gamma( \mathbf{s}^\star ) \geq (1-\alpha(\setbf{\tilde{S}}))[\gamma(\mathbf{s}) - \underbrace{\gamma(\mathbf{0})}_{= 0}] \, ,
\end{equation*} 
\vspace{-0.5em}
and $\alpha(\setbf{\tilde{S}}) \leq \alpha$.
\endgroup

For completeness we report the steps of \cite{roughgarden2015} to prove that $PoA_{CCE} \leq (1 + \alpha)$. Let $\mathbf{s}^\star = \argmax_{\mathbf{s}\in \setbf{S}}\gamma(\mathbf{s})$. Then, for any CCE $\sigma$ of $\set{G}$ we have 
\begingroup
\setlength{\abovedisplayskip}{3pt}
\setlength{\belowdisplayskip}{3pt}
\begin{align*}
\mathbb{E}_{\mathbf{s} \sim \sigma} [ \gamma(\mathbf{s})] & \geq \sum_{i=1}^N \mathbb{E}_{\mathbf{s} \sim \sigma} [ \pi_i(\mathbf{s})]  \geq \sum_{i=1}^N \mathbb{E}_{\mathbf{s} \sim \sigma} [ \pi_i(\mathbf{s}_i^\star, \mathbf{s}_{-i})] \\
& \geq \gamma( \mathbf{s}^\star )  - \alpha \: \mathbb{E}_{\mathbf{s} \sim \sigma} [ \gamma(\mathbf{s})] \, , 
\end{align*}
\endgroup
where the first inequality is due to condition iii) of valid utility games as per \refdef{def:valid_continuous_game}, the second inequality holds from $\sigma$ being a CCE, and the last one since $\set{G}$ is ($ 1,  \alpha$)-smooth. Moreover, linearity of expectation was used throughout. 
From the inequalities above it holds that for any CCE $\sigma$ of $\set{G}$, $\gamma(\mathbf{s}^\star)/ \mathbb{E}_{\mathbf{s} \sim \sigma} [ \gamma(\mathbf{s})] \leq 1 + \alpha$. Hence $PoA_{CCE} \leq 1 + \alpha$.

\begin{remark}Although \refthm{thm:main_theorem} requires DR-submodularity of $\gamma$ over $\reals{Nd}_+$ (for simplicity), only DR-submodularity over $\tilde{\setbf{S}}$ was used. In case $\gamma$ is DR-submodular only over $\setbf{S}$, one could consider  $\tilde{\gamma} : \reals{Nd}_+ \rightarrow \reals{}_+$ defined as $\tilde{\gamma}(\mathbf{s}) = \gamma( \min(\mathbf{s},\mathbf{s}_{max}))$ which is DR-submodular over $\reals{Nd}_+$. This can be proved using DR-submodularity and monotonicity of $\gamma$ over $\tilde{\setbf{S}}$. The same smoothness proof is obtained with $\tilde{\gamma}$ in place of $\gamma$ since the two functions are equal over $\setbf{S}$. However, the curvature of $\tilde{\gamma}$ with respect to $\tilde{\setbf{S}}$ is 1 and therefore a bound of 2 for $PoA_{CCE}$ is obtained.
\end{remark}
\vspace{-.5em}
\section{Extension to the non-submodular case}\label{sec:extension}
\vspace{-.5em}

In many applications \cite{bian2017b}, functions are close to being submodular, where this closedness has been measured in term of \emph{submodularity ratio} \cite{das2011} (for set functions) and \emph{weak-submodularity} \cite{hassani2017} (on continuous domains). Accordingly, in this section we relax condition i) of valid utility games (\refdef{def:valid_continuous_game}) and provide bounds for $PoA_{CCE}$ when the social function $\gamma$ is not necessarily DR-submodular. This case was never considered for the valid utility games of \cite{vetta2002,maehara2015}. 
We relax the \texttt{weak DR} property of $\gamma$ with the following definition. 
\begingroup
\setlength{\abovedisplayskip}{2pt}
\setlength{\belowdisplayskip}{2pt}
\vspace{-.5em}
\begin{definition}\label{def:playerwise-submodularity_ratio}
Given a game $\set{G}=(N, \{ \set{S}_i\}_{i=1}^N,  \{ \pi_i\}_{i=1}^N , \gamma)$ with $\gamma$ monotone,  we define \emph{generalized submodularity ratio} of $ \gamma$ as the largest scalar $\eta$ such that for any pair of outcomes $\mathbf{s}, \mathbf{s}' \in \setbf{S}$,
\begin{equation*}
 \sum\nolimits_{i=1}^N \gamma(\mathbf{s}_i + \mathbf{s}_i', \mathbf{s}_{-i}) - \gamma(\mathbf{s}) \geq \eta \big[ \gamma(\mathbf{s} + \mathbf{s}' ) - \gamma(\mathbf{s})  \big] \,. 
\end{equation*}
\end{definition}
 \endgroup 
It is straightforward to show that $\eta \in [0,1]$. Moreover, as stated in \refapp{B} (\refprop{prop:weak_DR_implies_playerwise_ratio_one}), if $\gamma$ is weakly DR-submodular then $\gamma$ has generalized submodularity ratio $\eta = 1$.
When strategies $\mathbf{s}_i$ are scalar (i.e., $d=1$), \refdef{def:playerwise-submodularity_ratio} generalizes the submodularity ratio by \cite{das2011} to continuous domains\footnote{In \refapp{B} we define an exact generalization of the submodularity ratio by \cite{das2011} to continuous domains. We relate it to \refdef{def:playerwise-submodularity_ratio} and compare it to the ratio by \cite{hassani2017}.}.

In addition, we relax the \texttt{DR} property of $\gamma$ as follows.
\begin{definition}\label{def:playerwise_DR}
Given a game $\set{G}=(N, \{ \set{S}_i\}_{i=1}^N,  \{ \pi_i\}_{i=1}^N , \gamma)$, we say that $ \gamma$ is \emph{playerwise DR-submodular} if  for every player $i$ and vector of strategies $\mathbf{s}_{-i}$, $\gamma(\cdot, \mathbf{s}_{-i})$  is DR-submodular. 
\end{definition}
\vspace{-0.5em}
Analogously to \refdef{def:DR_property}, if $\gamma$ is twice-differentiable, it is playwerwise DR-submodular iff for every $i \in [N]$
\vspace{-0.5em}
\begin{equation*}
\forall \mathbf{s} \in \setbf{S}, \quad \frac{\partial^2 \gamma(\mathbf{s})}{\partial [\mathbf{s}_i]_l \partial [\mathbf{s}_i]_m } \leq 0 ,\quad \forall \: l, m \in [d]  \, .
\end{equation*}
While \refdef{def:playerwise-submodularity_ratio} concerns the interactions among different players, \refdef{def:playerwise_DR} requires that $\gamma$ is DR-submodular with respect to each individual player.
When the social function $\gamma$ is DR-submodular, then it is also playerwise DR-submodular. Moreover, since the \texttt{DR} property is stronger than \texttt{weak DR}, $\gamma$ has generalized submodularity ratio $\eta = 1$. If $\gamma$ is not DR-submodular, however, the notions of \refdef{def:playerwise-submodularity_ratio} and \refdef{def:playerwise_DR} are not related. We visualize their differences in the following example. 
\begin{example}
 Consider a game with $N=2$, $d=2$, and $\gamma$ twice-differentiable. Let $\eta$ be the generalized submodularity ratio of $\gamma$. Assume the Hessian of $\gamma$ satisfies one of the three cases below, where with `$+$' or `$-$' we indicate the sign of its elements:
\begin{equation*}
\begin{gathered}
\tiny{
\begin{pmatrix}
    \tikzmarkin{a}(0.1,-0.1)(-0.1,0.2) - & - & - &- \\
    - &  -\tikzmarkend{a} & - & - \\
    - & - &   \tikzmarkin{b}(0.1,-0.1)(-0.1,0.2) -& -  \\
    - & - &- &- \tikzmarkend{b} 
\end{pmatrix}
    \, 
\begin{pmatrix}
     \tikzmarkin{c}(0.1,-0.1)(-0.1,0.2)  + & - & - &- \\
    - &  - \tikzmarkend{c}& - & - \\
    - & - &  \tikzmarkin{d}(0.1,-0.1)(-0.1,0.2) - & -  \\
    - & - &- & + \tikzmarkend{d}
\end{pmatrix}
    \,
\begin{pmatrix}
   \tikzmarkin{e}(0.1,-0.1)(-0.1,0.2) - & - & - &+ \\
    - &   -\tikzmarkend{e}& - & - \\
    - & - &  \tikzmarkin{f}(0.1,-0.1)(-0.1,0.2) -& -  \\
    + & - &- & - \tikzmarkend{f}
\end{pmatrix}
    } \\
   \text{1.} \hspace{.28\linewidth}\text{2.}  \hspace{.28\linewidth}\text{3.}
   \end{gathered}
   \end{equation*}
 \vspace{-1.5em}
 
From the previous definitions, the function $\gamma$ is playerwise DR-submodular iff all the entries highlighted in red are non-positive, while $\eta$ depends on all the off-diagonal entries. In case~1., all the entries are negative, hence $\gamma$ is DR-submodular. Thus, it is playerwise DR-submodular and has generalized submodularity ratio $\eta = 1$. In case~2., all off-diagonal entries are negative, hence $\gamma$ is weakly DR-submodular (see \refapp{app:properties_submodular}) and thus $\eta = 1$. However, $\gamma$ is not playerwise DR-submodular since some highlighted entries are positive. In case~3., $\gamma$ is playerwise DR-submodular and its generalized submodularity ratio depends on its parameters. 
\end{example}
Note that only case 1. of the previous example satisfies the conditions of \refthm{thm:main_theorem}. However, the following  \refthm{thm:weak_theorem} is applicable also to a subset of functions which fall in case 3. The proof can be found in \refapp{B}.

\vspace{-0.3em}
\begin{theorem}\label{thm:weak_theorem}
Let $\set{G}=(N, \{ \set{S}_i\}_{i=1}^N,  \{ \pi_i\}_{i=1}^N , \gamma)$ be a game where $\gamma$ is monotone, playerwise DR-submodular and has generalized submodularity ratio $\eta >0$. Then, if conditions ii) and iii) of \refdef{def:valid_continuous_game} are satisfied, $PoA_{CCE} \leq (1 +\eta)/\eta$.
\end{theorem}
\vspace{-.8em}
In light of the previous comments, when $\gamma$ is DR-submodular  \refthm{thm:weak_theorem} yields a bound of 2 which is always higher than $(1+\alpha)$ from \refthm{thm:main_theorem}. This is because the notion of curvature in \refdef{def:curvature} cannot be used in the more general setting of \refthm{thm:weak_theorem} since $\gamma$ may not be DR-submodular. 

In \refapp{B} we show that examples of functions with generalized submodularity ratio $1>\eta>0$ are products of monotone weakly DR-submodular functions and monotone affine functions. As a consequence, the following generalization of \refexa{example_2} falls into the set-up of \refthm{thm:weak_theorem}. 
\vspace{-0.5em}

\textbf{Sensor coverage problem with non-submodular objective}.
 Consider the sensor coverage problem defined in \refexa{example_2}, where the weights $w_r$'s are monotone affine functions $w_r : \reals{Nd}_+ \rightarrow \reals{}_+$ rather than constants. For instance, the probability that an event occurs in location $r$ can increase with the average amount of energy allocated to that location. That is, $\gamma(\mathbf{x}) = \sum_{r \in [d]}w_r(\mathbf{x}) \: P(r,\mathbf{x})$ with $w_r(\mathbf{x}) = \mathbf{a}_r\: \frac{\sum_{i=1}^N[\mathbf{x}_i]_r}{N} + b_r$. To maximize $\gamma$ one could set up a game $\set{G}$ where condition ii) of \refdef{def:valid_continuous_game} is satisfied with equality, as shown in \refsec{sec:distributed_maximization}.  In \refapp{sec:extension_example} we show that $\gamma$ has generalized submodularity ratio $1 > \eta > 0$, it is playerwise DR-submodular, and that $\gamma(\mathbf{x}) \geq \frac{1}{2} \sum_{i=1}^N \pi_i(\mathbf{x})$  for every $\mathbf{x}$, which is a weaker version of condition iii). Nevertheless, using \refthm{thm:weak_theorem} and the last proof steps of \refsec{sec:proof_of_thm1} we prove that $PoA_{CCE} \leq (1 + 0.5\eta)/0.5\eta$. We also show that $\gamma$ is concave in each $\set{X}_i$. Therefore a distributed implementation of online gradient ascent maximizes $\gamma$ up to  $0.5\eta / (1 + 0.5\eta)$ approximations.
 
 \vspace{0.5em}
We remark that our definitions of curvature, submodularity ratio, and Theorems \ref{thm:main_theorem}-\ref{thm:weak_theorem} can also be applied to games and optimizations over integer domains, i.e., when $\set{S}_i \subseteq \mathbb{Z}^d_+$ and $\gamma$ is defined on integer lattices. 
\vspace{-.5em}
\section{Experimental results}\label{sec:examples}
\vspace{-1em}

In this section we analyze the examples defined in \refsec{sec:problem_formulation} using the developed framework. 

\begin{figure*}[!ht]
\centering
\begin{subfigure}{.32\linewidth}
\centering
  \includegraphics[width=1.1\linewidth]{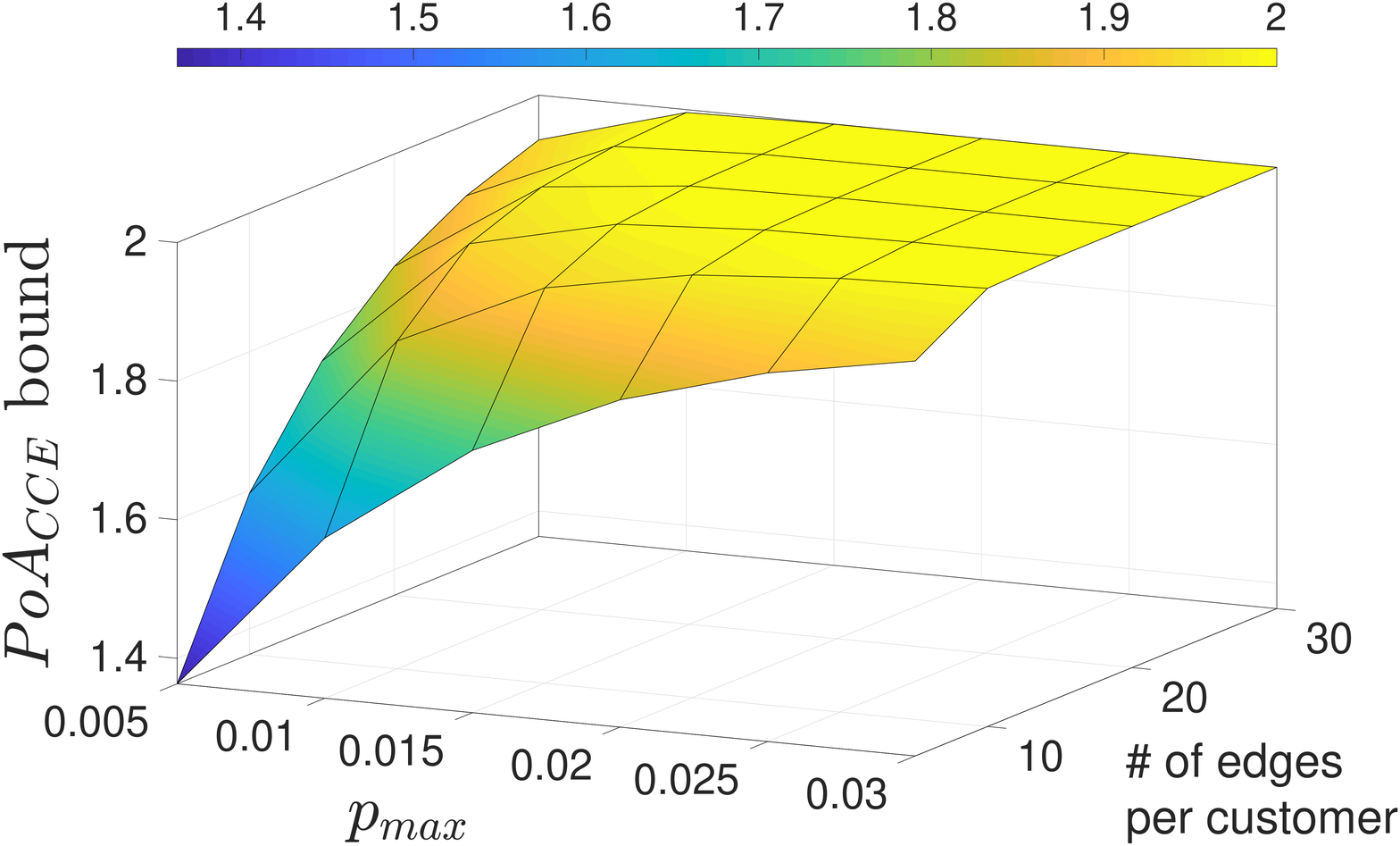}
  \caption{Budget allocation game}
  \label{fig:budget_alloc}
\end{subfigure}
$\quad$
\begin{subfigure}{.64\linewidth}
\centering
  \includegraphics[width=.49\linewidth]{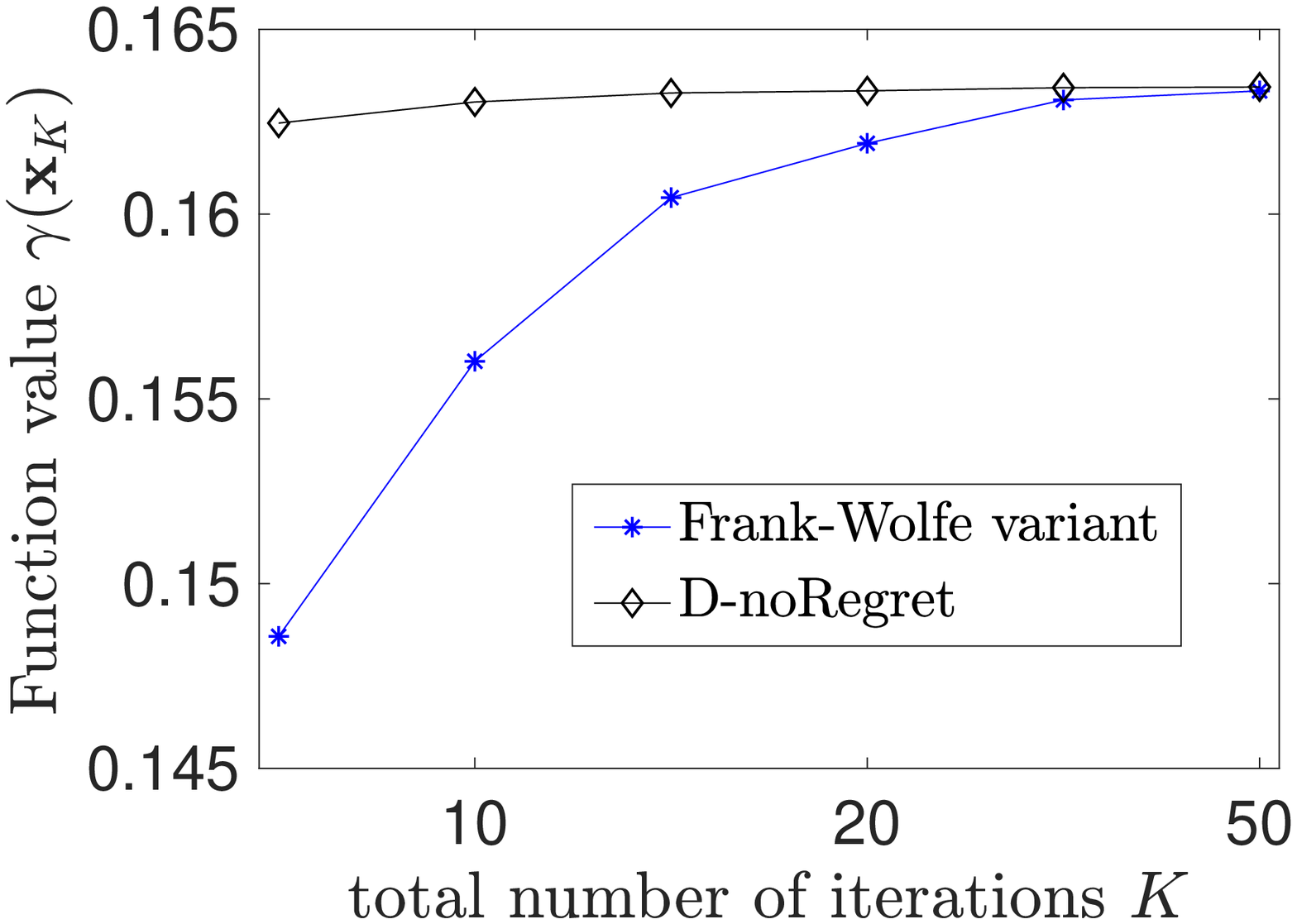}
  \includegraphics[width=.49\linewidth]{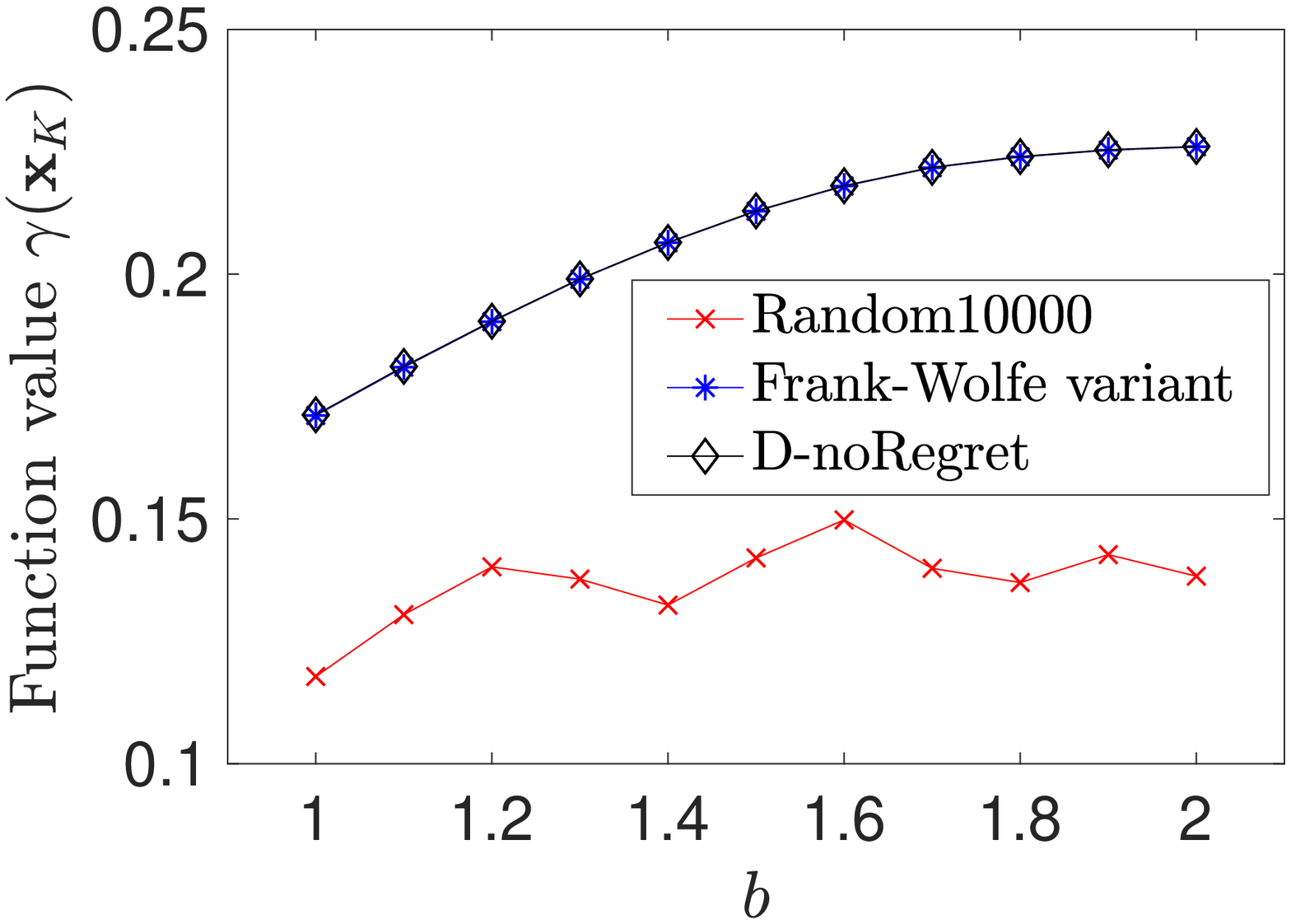}
  \caption{Sensor coverage problem}
  \label{fig:sensor_cov}
\end{subfigure}
\caption{a) Bounds for $PoA_{CCE}$, as a function of $p_{max}$ and the number of edges connected to each customer. The bounds strictly improve the bound of 2 provided by \cite{maehara2015} which does not depend on any of such parameters. b) Performance of \textsc{Frank-Wolfe} variant and \textsc{D-noRegret} for $K=3000$ iterations. Left: $\gamma(\mathbf{x}_K)$ as function of $K$. Right: $\gamma(\mathbf{x}_K)$ as function of the budgets $b$, with $K = 3000$. \textsc{D-noRegret} shows faster convergence, but for $K=3000$ the two algorithms perform equally.  }
\label{fig:both_examples}
\end{figure*}

\vspace{-0.8em}
\subsection{Continuous budget allocation game}
\vspace{-0.8em}

We consider $N = 10$ advertisers in a market with $d = 100$ channels and $\vert \set{T} \vert = 10'000$ customers. For the budget constraints we select $b_i=1$, $\bar{\mathbf{s}}_i = \mathds{1}$ and each entry of $\mathbf{c}_i$ is sampled uniformly at random from $[0,1]$. For each $i\in [N], r\in \set{R}, t \in \set{T}$, $p_i(r,t)$ is drawn uniformly at random in $[0.8 , 1]\:p_{max}$. In \reffig{fig:budget_alloc} we visualize the bound for $PoA_{CCE}$ obtained in \refprop{prop:Example_1} for different values of $p_{max}$ and the number of random edges connected to each  customer. The chosen ranges ensure that a sufficient fraction of customers will be activated. For instance, for $p_{max}=0.01$ and drawing $20$ random edges for each customer, we obtained
\footnote{Since $\gamma$ is monotone DR-submodular and $\setbf{S}$ is down-closed, we used the \textsc{Frank-Wolfe} variant by \cite{bian2017a} to maximize $\gamma$ up to $(1-e^{-1})$ approximations.}
an expected number of 2270 activated customers. This is in line with \cite{maehara2015}, where problem parameters were chosen such that 
$\frac{1}{5}$ of the customers are activated. As visible, the bound decreases when the activation probabilities decrease or when less edges are connected to each customer and can strictly improve the bound of $2$ provided by \cite{maehara2015}.



\vspace{-0.5em}
\subsection{Sensor coverage with continous assignments}
\vspace{-0.5em}

To maximize the probability $\gamma$ of detecting an event, we compare the performance of \textsc{D-noRegret} with the \textsc{Frank-Wolfe} variant by \cite{bian2017a} and a hit-and-run sampler \cite{kroese2013} \textsc{Random10000} which samples 10'000 random feasible points. We choose $N = 5$ sensors and $d = 30$ locations. For each budget constraint $\set{X}_i$, entries of $\mathbf{c}_i$ are chosen uniformly at random in $\frac{1}{d}[1,3]$, $b_i = 1$, and $\bar{\mathbf{x}}_i = \mathds{1}$. For each $i\in [N]$ and $r\in [d]$, we select $p_i^r = p = 0.05$ and $w_r$ uniformly at random in $[0, 1]$ such that $ \sum_{r=1}^d w_r= 1$. Under this choice, $\alpha \approx 0.4$, and \textsc{D-noRegret} has approximation guarantee $\approx 0.71$ which is greater than $(1- e^{-1}) \approx 0.63$. We initialize both \textsc{Frank-Wolfe} and \textsc{D-noRegret} at $\mathbf{x}_0 = 0$ and run them for $K \in \{10,20,50,100,500, 3000 \}$ iterations. Since the constraints are decoupled, also the \textsc{Frank-Wolfe} variant can be implemented distributively. Step sizes for both algorithms are chosen costant and proportional to $1/K$ and $1/\sqrt{K}$ as per \cite{bian2017a} and \cite[Lemma 3.2]{flaxman2005}, respectively. In \reffig{fig:sensor_cov} we compare the values of $\gamma(\mathbf{x}_K)$  as a function of the number of iterations $K$ (left plot). Moreover, for $K = 3000$, we compare the performance of the algorithms when we enlarge the constraints $\set{X}_i$ by choosing $b_i = b$ for each $i$, with  $b \in  \{1,1.1, 1.2,\ldots, 2\}$ (right plot). As visible, \textsc{D-noRegret} shows faster convergence than \textsc{Frank-Wolfe} variant. However, for $K=3000$ the two algorithms return the same values. Average computation times per iteration are $0.019~\mathrm{s}$ and $0.009~\mathrm{s}$ for \textsc{Frank-Wolfe} and \textsc{D-noRegret}, respectively, on a 16 Gb machine at 3.1 GHz using Matlab.


\vspace{-.5em}
\section{Conclusions and future work}
\vspace{-0.9em}

We bounded the robust price of anarchy for a subclass of continuous games, denoted as valid utility games with continuous strategies. Our bound relies on a particular structure of the game and on the social function being monotone DR-submodular. We introduced the notion of curvature of a monotone DR-submodular function and refined the bound using this notion. In addition, we extended the obtained bounds to a class of non-submodular functions.  
We showed that valid utility games can be designed to maximize monotone DR-submodular functions subject to disjoint constraints. For a subclass of such functions, our approximation guarantees improve the ones in the literature. 
We demonstrated our results numerically via a continuous budget allocation game and a sensor coverage problem.
In light of the obtained approximation guarantees, we believe that the introduced notion of curvature of a monotone DR-submodular function can be used to tighten existing guarantees for constrained maximization. Currently, we are studying the tightness of the obtained bounds and their applicability to several continuous games such as auctions. 

%

\bibliographystyle{plain} 
\bibliography{../../My_Biblio}

\clearpage
\appendix
\section{Supplementary material for Sections~\ref{sec:main_results}-\ref{sec:background}}\label{A}

\subsection{Proof of \refrem{remark:curvature_in_zero_one}}
Since $f$ is monotone, $f(\mathbf{x}+ k\mathbf{e}_i) \geq f(\mathbf{x})$ and $f(k\mathbf{e}_i) \geq f(\mathbf{0})$ for any $\mathbf{x} \in  \setbf{Z}$, $i \in [n]$, and $k \in \reals{}_+$. Hence, $\alpha(\setbf{Z}) \leq 1$. Moreover,  $\inf_{\substack{\mathbf{x}\in \setbf{Z} \\ i \in [n]}}{\lim_{k\rightarrow 0^+}\frac{f(\mathbf{x}+ k \mathbf{e}_i) - f(\mathbf{x})}{f(k \mathbf{e}_i)- f(\mathbf{0})}}   \leq 1$, since the considered ratio equals 1 when $\mathbf{x} = \mathbf{0}$. Hence, $\alpha(\setbf{Z}) \geq 0$. \hfill \qed

\subsection{Proof of \refrem{remark:bound_2}}
The proof is obtained simply noting that the curvarture $\alpha(\tilde{\setbf{S}})$ of $\gamma$ is always upper bounded by 1.
 \hfill \qed

\subsection{Proof of \refprop{prop:Example_1}}
We first show that the budget allocation game of Example \ref{example_1} is a valid utility game with continuous strategies. In fact, for any $l \in [Nd]$
\begingroup
\setlength{\abovedisplayskip}{3pt}
\setlength{\belowdisplayskip}{3pt} 
\begin{equation*}
[\nabla \gamma(\mathbf{s})]_l = \sum_{ t\in \set{T} : m \in \Gamma(t)} - \ln (1 - p_j(m,t))  \prod_{i = 1}^N(1-P_i(\mathbf{s}_i,t)) \, ,
\end{equation*}
\endgroup
where $j\in [N]$ and $m \in [d]$ are the indexes of advertiser and channel corresponding to coordinate $l \in [Nd]$, respectively. Hence, $\gamma$ is monotone since $[\nabla \gamma(\mathbf{s})]_l \geq 0$ for any $l \in [Nd]$ and $\mathbf{s} \in \reals{Nd}_+$. Moreover, $\gamma$ is DR-submodular since $\gamma(\mathbf{s})= \sum_{t \in \set{T}} \gamma_t(\mathbf{s})$ where $\gamma_t(\mathbf{s}) =  1 - \prod_{i=1}^N(1- P_i(\mathbf{s}_i,t))$ is such that for any $j,l \in [N],m,n \in [d]$, $\frac{\partial^2 \gamma_t(\mathbf{s})}{\partial [\mathbf{s}_j]_m \partial [\mathbf{s}_l]_n}=  -\ln(1 - p_j(m,t) \ln(1- p_l(n,t))  \prod_{i = 1}^N(1-P_i(\mathbf{s}_i,t)) \leq 0$ for any $\mathbf{s} \in \reals{Nd}_+$. Finally, condition ii) can be verified equivalently as in \cite[Proof of Proposition 5]{maehara2015} and condition iii) holds with equality. 

The set $\setbf{\tilde{S}} :=\{ \mathbf{x} \in \reals{Nd}_+ \mid \mathbf{0} \leq \mathbf{x} \leq  \mathbf{s}_{max} \}$ with $\mathbf{s}_{max} = 2(\bar{s}_1,\ldots,  \bar{s}_N)$ is such that $\mathbf{s} + \mathbf{s}' \leq \mathbf{s}_{max}$ for any pair $\mathbf{s},\mathbf{s}' \in \setbf{S}$. Moreover, using the expression of $\nabla \gamma(\mathbf{s})$, the curvature of $\gamma$ with respect to $\setbf{\tilde{S}}$ is  
\begingroup
\setlength{\abovedisplayskip}{6pt}
\setlength{\belowdisplayskip}{6pt} 
\begin{align*}
& 1-\alpha( \tilde{\setbf{S}}) = \inf\nolimits_{\substack{\mathbf{s}\in \tilde{\setbf{S}} \\ l \in [Nd]}}\frac{[\nabla \gamma(\mathbf{s})]_l}{[\nabla \gamma(0)]_l} = \\
&  \min_{ i \in [N], r\in [d]} \frac{  \sum\limits_{t \in \set{T}: r \in \Gamma(t)} \ln(1-p_i(r,t))\prod\limits_{j \in [N]}(1-P_j(2\bar{\mathbf{s}}_j,t))}{ \sum\limits_{t \in \set{T}: r \in \Gamma(t)} \ln(1-p_i(r,t)) } \\
& =:  1 -\alpha > 0 .
 \end{align*}
\endgroup 
Hence, using \refthm{thm:main_theorem} we conclude that $PoA_{CCE} \leq 1 + \alpha$.  \hfill \qed

\subsection{Proof of Fact 1}
Condition i) holds since $\gamma$ is monotone DR-submodular by definition. Also, condition ii) holds with equality. Moreover, defining (with abuse of notation) $[\mathbf{s}]_1^i = (\mathbf{s}_1, \ldots, \mathbf{s}_i, \mathbf{0}, \ldots, \mathbf{0})$ for $i\in [N]$ with $[\mathbf{s}]_1^0 = \mathbf{0}$, condition iii) holds since by DR-submodularity one can verify that $\sum_{i=1}^N \hat{\pi}_i(\mathbf{s}) =  \sum_{i=1}^N \gamma(\mathbf{s}) - \gamma(\mathbf{0}, \mathbf{s}_{-i}) \leq \gamma([\mathbf{s}]_1^i) - \gamma([\mathbf{s}]_1^{i-1}) = \gamma(\mathbf{x}) - \gamma(\mathbf{0}) = \gamma(\mathbf{x})$. \hfill \qed

\subsection{Proof of \refcor{cor:D_no_regret}}
By definition of $\alpha$, and according to \refthm{thm:main_theorem}, $\hat{\set{G}}$ is such that $PoA_{CCE} \leq (1 + \alpha)$. In other words, letting $\mathbf{s}^\star = \argmax_{\mathbf{s}\in \setbf{S}}\gamma(\mathbf{s})$, any CCE $\sigma$ of $\hat{\set{G}}$ satisfies 
$\mathbb{E}_{\mathbf{s} \sim \sigma} [ \gamma(\mathbf{s})] \geq 1/( 1 + \alpha)\gamma(\mathbf{s}^\star)$. Moreover, since players simultaneously use no-regret algorithms \textsc{D-noRegret} converges to one of such CCE \cite{gordon2008,roughgarden2015}. Hence, the statement of the remark follows.  \hfill \qed

\subsection{Proof of \refprop{prop:Example_2}}
Consider the sensor coverage problem with continuous assignments defined in Example \ref{example_2}. 
We first show that $\gamma$ is a monotone DR-submodular function. In fact, for any $i \in [Nd]$, $[\nabla \gamma(\mathbf{x})]_i = -\ln(1-p_l^m) \prod_{i \in [N]}(1-p_i^m)^{[\mathbf{x}_i]_m} \geq 0$, where $l$ and $m$ and the indexes of sensor and location corresponding to coordinate $i$, respectively. Moreover, for any pair of sensors $j,l \in [N]$, $\frac{\partial^2 \gamma(\mathbf{x})}{\partial [\mathbf{x}_j]_m \partial [\mathbf{x}_l]_n}= 
-\ln(1-p_j^m) \ln(1-p_l^n) \prod_{i \in [N]}(1-p_i^m)^{[\mathbf{x}_i]_m}  \leq 0$ if  $m = n$, and $0$ otherwise. 
The problem of maximizing $\gamma$ subject to $\setbf{X} = \prod_{i=1}^N\set{X}_i$, hence, is one of maximizing a monotone DR-submodular function subject to decoupled constraints discussed in \refsec{sec:distributed_maximization}. Thus, as outlined in \refsec{sec:distributed_maximization}, we can set-up a valid utility game $\hat{\set{G}}$.

The vector $ \mathbf{x}_{max} = 2\bar{\mathbf{x}} = 2(\bar{\mathbf{x}}_1, \ldots, \bar{\mathbf{x}}_N)$ is such that   $\forall \mathbf{x},\mathbf{x}'\in \setbf{X}$, $\mathbf{x}+\mathbf{x}' \leq   \mathbf{x}_{max}$. Moreover, defining $\setbf{\tilde{X}} :=\{ \mathbf{x} \in \reals{Nd}_+ \mid \mathbf{0} \leq \mathbf{x} \leq  \mathbf{x}_{max} \}$, the curvature of $\gamma$ with respect to $\tilde{\setbf{X}}$, satisfies $\alpha( \tilde{\setbf{X}}) = 1 - \inf_{\substack{\mathbf{x}\in \tilde{\setbf{X}} \\ l \in [Nd]}}\frac{[\nabla \gamma(\mathbf{x})]_l}{[\nabla \gamma(0)]_l} = 1 - \min_{r \in [d]} \prod_{i \in [N]}(1-p_i^r)^{2\bar{\mathbf{x}}_i}  = \max_{r \in [d]}P(r, 2\bar{\mathbf{x}}) = \alpha$.
Hence, by \refcor{cor:D_no_regret}, any no-regret distributed algorithm has expected approximation ratio of $1/(1+ \alpha)$. In addition, $\gamma$ is also concave in each $\set{X}_i$, since the $(d\times d)$ blocks on the diagonal of its Hessian are diagonal and negative, hence online gradient ascent ensures no-regret for each player \cite{flaxman2005} and can be run in a distributed manner. \hfill \qed

\subsection{Properties of (twice) differentiable submodular functions}\label{app:properties_submodular}
As mentioned in \refsec{sec:background}, submodular continuous functions are defined on subsets of $\reals{n}$ of the form $\setbf{X} = \prod_{i=1}^n \set{X}_i$, where each $\set{X}_i$ is a compact subset of $\reals{}$. From the \texttt{weak DR} property (\refdef{def:weak_DR_property}) it follows that, when $f$ is differentiable, it is submodular iff 
$$\forall \mathbf{x}, \mathbf{y} \in \setbf{X} : \mathbf{x} \leq \mathbf{y}, \forall i \text{ s.t. } x_i = y_i , \: \nabla_i f(\mathbf{x} ) \geq \nabla_i f(\mathbf{y}) \,.$$
That is, the gradient of $f$ is a weak antitone mapping from $\mathbb{R}^n$ to $\mathbb{R}^n$. 

Moreover, we saw that a function $f: \setbf{X} \rightarrow \reals{}$ is \emph{submodular} iff
for all $\mathbf{x} \in \setbf{X}$, $\forall i \neq j$ and $a_i,a_j >0$ s.t. $x_i + a_i \in \mathcal{X}_i$, $x_j + a_j \in \mathcal{X}_j$, we have \cite{bach2015}
\begin{equation*}
f(\mathbf{x} + a_i \mathbf{e}_i ) - f(\mathbf{x})  \geq    f(\mathbf{x} + a_i \mathbf{e}_i+  a_j \mathbf{e}_j) - f(\mathbf{x} + a_j \mathbf{e}_j) \,.
\end{equation*}
As visible from the latter condition, when $f$ is twice-differentiable, it is submodular iff all the off-diagonal entries of its Hessian are non-positive \cite{bach2015}:
\begin{equation*}
\forall \mathbf{x} \in \setbf{X}, \quad \frac{\partial^2 f(\mathbf{x})}{\partial x_i \partial x_j} \leq 0 ,\quad \forall i \neq j  \, .
\end{equation*}
Hence, the class of submodular continuous functions contains a subset of both convex and concave functions.

Similarly, from the \texttt{DR} property (\refdef{def:DR_property}) it follows that for a differentiable continuous function DR-submodularity is equivalent to 
$$\forall \mathbf{x} \leq \mathbf{y}, \nabla f(\mathbf{x} ) \geq \nabla f(\mathbf{y}) \,.$$
That is, the gradient of $f$ is an antitone mapping from $\mathbb{R}^n$ to $\mathbb{R}^n$. More precisely, \cite[Proposition 2]{bian2017a} showed that a function $f$ is DR-submodular iff it is submodular (weakly DR-submodular) and \emph{coordinate-wise concave}. A function $f: \setbf{X} \rightarrow \mathbb{R}$ is coordinate-wise concave if, for all $\mathbf{x} \in \setbf{X}$, $\forall i \in [n], \forall k,l \in \mathbb{R}_+$ s.t. $(\mathbf{x} + k \mathbf{e}_i)$, $(\mathbf{x} + l \mathbf{e}_i)$, and $(\mathbf{x} + (k+l) \mathbf{e}_i)$ are in $ \setbf{X}$, we have
\begin{equation*}
f(\mathbf{x} + k \mathbf{e}_i ) - f(\mathbf{x})  \geq    f(\mathbf{x} + (k+l) \mathbf{e}_i) - f(\mathbf{x} + l \mathbf{e}_i) \,,
\end{equation*}
or equivalently, if twice differentiable, $\frac{\partial^2 f(\mathbf{x})}{\partial x_i^2} \leq 0$ $\forall i \in [n]$.
Hence, as stated in \refsec{sec:main_results}, a twice-differentiable function is DR-submodular iff all the entries of its Hessian are non-positive:
\begin{equation*}
\forall \mathbf{x} \in \setbf{X}, \quad \frac{\partial^2 f(\mathbf{x})}{\partial x_i \partial x_j} \leq 0 ,\quad \forall i, j  \, .
\end{equation*}
\subsection{Proof of \refprop{prop:group_weak_DR_property}}
(property of \refprop{prop:group_weak_DR_property}  $\rightarrow$ \texttt{weak DR}) \\
We want to prove that for all $\mathbf{x} \leq \mathbf{y} \in \setbf{X}$, $\forall i$ s.t. $x_i = y_i$, $\forall k \in \mathbb{R}_+$ s.t. $(\mathbf{x} + k \mathbf{e}_i)$ and $(\mathbf{y} + k \mathbf{e}_i)$ are in $\setbf{X}$,
\begin{equation*}
f(\mathbf{x} + k \mathbf{e}_i ) - f(\mathbf{x})  \geq    f(\mathbf{y} + k \mathbf{e}_i) - f(\mathbf{y}) \,.
\end{equation*}
This is trivially done choosing $\mathbf{z} =  k\mathbf{e}_i$. Note that $\mathbf{z}$ is such that $z_i = 0,  \forall i \in \{ i | y_i > x_i\}$, so the property of \refprop{prop:group_weak_DR_property} can indeed be applied. 

(\texttt{weak DR} $\rightarrow$ property of \refprop{prop:group_weak_DR_property}) \\
For all $\mathbf{x} \leq \mathbf{y} \in  \setbf{X}$, $\forall \mathbf{z} \in \reals{n}_+$ s.t. $(\mathbf{x} +  \mathbf{z})$ and $(\mathbf{y} +  \mathbf{z})$ are in $\setbf{X}$, with $z_i = 0$ $\forall i \in [n]:   y_i > x_i $, we have
\begin{align*}
f(\mathbf{x} + \mathbf{z}) & - f(\mathbf{x}) = \sum_{i =1}^n f(\mathbf{x} + [\mathbf{z}]_{1}^i) - f(\mathbf{x} + [\mathbf{z}]_{1}^{i-1}) \\
& = \sum_{i : x_i = y_i } f(\mathbf{x} +[\mathbf{z}]_{1}^{i-1}+ z_i \mathbf{e}_i ) - f(\mathbf{x} + [\mathbf{z}]_{1}^{i-1}) \\
& \geq \sum_{i : x_i = y_i } f(\mathbf{y} +[\mathbf{z}]_{1}^{i-1}+ z_i \mathbf{e}_i ) - f(\mathbf{y} + [\mathbf{z}]_{1}^{i-1}) \\
& = \sum_{i =1}^n f(\mathbf{y} + [\mathbf{z}]_{1}^i) - f(\mathbf{y} + [\mathbf{z}]_{1}^{i-1}) \\
& = f(\mathbf{y} + \mathbf{z}) - f(\mathbf{y})  \,.
\end{align*}
 The first equality is obtained from a telescoping sum, the second equality follows since when $y_i > x_i$, $z_i = 0$. The inequality follows from \texttt{weak DR} property of $f$ and the last two equalities are similar to the first two.  \hfill \qed 
 
\noindent
\subsection{Proof of \refprop{prop:group_DR_property}}
(property of \refprop{prop:group_DR_property} $\rightarrow$ \texttt{DR}) \\
We want to prove that for all $\mathbf{x} \leq \mathbf{y} \in \setbf{X}$, $\forall i \in [n]$, $\forall k \in \mathbb{R}_+$ s.t. $(\mathbf{x} + k \mathbf{e}_i)$ and $(\mathbf{y} + k \mathbf{e}_i)$ are in $\setbf{X}$,
\begin{equation*}
f(\mathbf{x} + k \mathbf{e}_i ) - f(\mathbf{x})  \geq    f(\mathbf{y} + k \mathbf{e}_i) - f(\mathbf{y}) \,.
\end{equation*}
This is trivially done choosing $\mathbf{z} =  k\mathbf{e}_i$ and applying the property of \refprop{prop:group_DR_property}. 

(\texttt{DR} $\rightarrow$ property of \refprop{prop:group_DR_property}) \\
For all $\mathbf{x} \leq \mathbf{y} \in  \setbf{X}$, $\forall \mathbf{z} \in \reals{n}_+$ s.t. $(\mathbf{x} +  \mathbf{z})$ and $(\mathbf{y} +  \mathbf{z})$ are in $\setbf{X}$, we have
\begin{align*}
f(\mathbf{x} + \mathbf{z}) & - f(\mathbf{x}) = \sum_{i =1}^n f(\mathbf{x} + [\mathbf{z}]_{1}^i) - f(\mathbf{x} + [\mathbf{z}]_{1}^{i-1}) \\
& =\sum_{i =1}^n f(\mathbf{x} +[\mathbf{z}]_{1}^{i-1}+ z_i \mathbf{e}_i ) - f(\mathbf{x} + [\mathbf{z}]_{1}^{i-1}) \\
& \geq \sum_{i =1}^n f(\mathbf{y} +[\mathbf{z}]_{1}^{i-1}+ z_i \mathbf{e}_i ) - f(\mathbf{y} + [\mathbf{z}]_{1}^{i-1}) \\
& = \sum_{i =1}^n f(\mathbf{y} + [\mathbf{z}]_{1}^i) - f(\mathbf{y} + [\mathbf{z}]_{1}^{i-1}) \\
& = f(\mathbf{y} + \mathbf{z}) - f(\mathbf{y})  \,.
\end{align*}
 The first  and last equalities are telescoping sums and the inequality follows from the \texttt{DR} property of $f$.  \hfill \qed 

\subsection{Proof of \refprop{prop:equivalence_curvature}}
By \refdef{def:curvature}, the curvature $\alpha(\setbf{Z})$ of $f$ w.r.t. $\setbf{Z}$ satisfies
\begin{equation}\label{aux:condition}
f(\mathbf{x}+ k \mathbf{e}_i) - f(\mathbf{x}) \geq (1-\alpha(\setbf{Z})) [f(k \mathbf{e}_i)- f(\mathbf{0})] \, ,
\end{equation} 
for any $\mathbf{x}\in \setbf{Z} , i \in [n]$ s.t. $\mathbf{x} +  k \mathbf{e}_i \in \setbf{Z}$ with $k \rightarrow 0_+$. We firstly show that condition \eqref{aux:condition} indeed holds for any $\mathbf{x}\in \setbf{Z} , i \in [n]$, and $k \in \reals{}_+$ s.t. $\mathbf{x} +  k \mathbf{e}_i \in \setbf{Z}$, by using monotonicity and coordinate-wise concavity of $f$. As seen in \refapp{app:properties_submodular}, DR-submodularity implies coordinate-wise concavity.  To this end, we define 
\begin{equation*}
  \alpha_i^k( \setbf{Z}) = 1 - \inf_{\substack{\mathbf{x}\in \setbf{Z} :\\\mathbf{x}+ k \mathbf{e}_i \in \setbf{Z}}}\frac{f(\mathbf{x}+ k \mathbf{e}_i) - f(\mathbf{x})}{f(k \mathbf{e}_i)- f(\mathbf{0})}   \,.
 \end{equation*}
Hence, it sufficies to prove that, for any $i \in [n]$, $\alpha_i^k( \setbf{Z})$ is non-increasing in $k$.
Note that by DR-submodularity, 
\begin{equation*}
  \alpha_i^k( \setbf{Z}) = 1 - \frac{f(\mathbf{z}_{max}) - f(\mathbf{z}_{max} - k \mathbf{e}_i)}{f(k \mathbf{e}_i)- f(\mathbf{0})}   \,.
 \end{equation*}
 Hence, for any pair $l,m \in \reals{}_+$ with $l < m$,  $\alpha_i^m( \setbf{Z}) \geq  \alpha_i^l( \setbf{Z})$ is true whenever 
 \begin{equation*}
  \frac{f(\mathbf{z}_{max}) - f(\mathbf{z}_{max} - m \mathbf{e}_i)}{f(m \mathbf{e}_i)- f(\mathbf{0})} \geq \frac{f(\mathbf{z}_{max}) - f(\mathbf{z}_{max} - l \mathbf{e}_i)}{f(l \mathbf{e}_i)- f(\mathbf{0})} 
  \,.
 \end{equation*}
The last inequality is satisfied since, by coordinate-wise concavity, 
 $[f(\mathbf{z}_{max}) - f(\mathbf{z}_{max} - m \mathbf{e}_i)]/m  \geq[f(\mathbf{z}_{max}) - f(\mathbf{z}_{max} - l \mathbf{e}_i)]/l$
 and $[f(m \mathbf{e}_i) - f(	\mathbf{0})]/m  \leq  [f( l \mathbf{e}_i) - f(\mathbf{0})]/l$. This is because, given a concave function $g:\reals{} \rightarrow \reals{}$, the quantity
 \begin{equation*}
 R(x_1, x_2):= \frac{g(x_2) - g(x_1)}{x_2 - x_1}
\end{equation*}  
 is non-increasing in $x_1$ for fixed $x_2$, and vice versa. 
Moreover, monotonicity ensures that all of the above ratios are non-negative.


To conclude the proof of \refprop{prop:equivalence_curvature} we show that if condition \eqref{aux:condition} holds for any $\mathbf{x}\in \setbf{Z} , i \in [n]$, and $k \in \reals{}_+$ s.t. $\mathbf{x} +  k \mathbf{e}_i \in \setbf{Z}$, then the result of the proposition follows. Indeed, for any $\mathbf{x}, \mathbf{y}$ s.t. $\mathbf{x}+ \mathbf{y} \in \setbf{Z}$ we have
\begin{align*}
f(\mathbf{x} + \mathbf{y}) & - f(\mathbf{x}) = \sum_{i =1}^n f(\mathbf{x} + [\mathbf{y}]_{1}^i) - f(\mathbf{x} + [\mathbf{y}]_{1}^{i-1}) \\
& = \sum_{i =1}^n f(\mathbf{x} +[\mathbf{y}]_{1}^{i-1}+ y_i \mathbf{e}_i ) - f(\mathbf{x} + [\mathbf{y}]_{1}^{i-1}) \\
& \geq (1-\alpha(\setbf{Z})) \sum_{i =1}^n f(y_i \mathbf{e}_i) - f(\mathbf{0}) \\
& \geq (1-\alpha(\setbf{Z})) \sum_{i =1}^n f([\mathbf{y}]_{1}^i) - f([\mathbf{y}]_{1}^{i-1}) \\ 
& = (1-\alpha(\setbf{Z}))  (f(\mathbf{y}) - f(\mathbf{0}))  \, ,
\end{align*}
where the first inequality follows by condition \eqref{aux:condition} and the second one from $f$ being weakly DR-submodular (and using \refprop{prop:group_weak_DR_property}). \hfill \qed


 \clearpage
\section{Supplementary material for \refsec{sec:extension}}\label{B}

In the first part of this appendix we generalize the submodularity ratio defined in \cite{das2011} for set functions to continuous domains and discuss its main properties. We compare it to the ratio by \cite{hassani2017} and we relate it to the generalized submodularity ratio defined in \refdef{def:playerwise-submodularity_ratio}. Then, we provide a class of social functions with generalized submodularity ratio $0< \eta <1$ and we report the proof of \refthm{thm:weak_theorem}. Finally, we analyze the sensor coverage problem with the non-submodular objective defined in \refsec{sec:extension}.

\subsection{Submodularity ratio of a monotone function on continuous domains}

We generalize the class of submodular continuous functions, defining the \emph{submodularity ratio} $\eta \in [0,1]$ of a monotone function defined on a continuous domain. 

\begin{definition}[submodularity ratio]\label{def:weak_DR_sub_ratio}
  The submodularity ratio of a monotone function $f: \setbf{X} \subseteq \reals{n}_+ \rightarrow \mathbb{R}$ is the largest scalar $\eta$ such that  for all $\mathbf{x}, \mathbf{y} \in \setbf{X}$ such that $\mathbf{x} + \mathbf{y} \in \setbf{X}$,
\begin{equation*}
\sum_{i=1}^n \big[ f(\mathbf{x} + y_i \mathbf{e}_i) - f(\mathbf{x}) \big] \geq \eta \big[ f(\mathbf{x} +  \mathbf{y}) - f(\mathbf{x}) \big] \,. 
\end{equation*}
\end{definition} 

It is straightforward to show that $\eta \in [0,1]$ and, when restricted to binary sets $\setbf{X} = \{0,1\}^n$, \refdef{def:weak_DR_sub_ratio} coincides with the submodularity ratio defined in \cite{das2011} for set functions. A set function is submodular iff it has submodularity ratio $\eta = 1$ \cite{das2011}. However, functions with submodularity ratio $0 < \eta <1$ still preserve `nice' properties in term of maximization guarantees. Similarly to \cite{das2011}, we can affirm the following. 

\begin{proposition}\label{prop:weak_submodularity_ratio}
A function $f: \setbf{X} \subseteq \reals{n}_+ \rightarrow \mathbb{R}$ is weakly DR-submodular (\refdef{def:weak_DR_property}) iff it has submodularity ratio $\eta = 1$.
\end{proposition}
\begin{proof}
If $f$ is weakly DR-submodular (\refdef{def:weak_DR_property}), then for any $\mathbf{x}, \mathbf{y} \in \setbf{X}$, 
\begin{align*}
& \sum_{i=1}^d  f(\mathbf{x} + y_i \mathbf{e}_i) - f(\mathbf{x})\\
&  \geq \sum_{i=1}^d  f(\mathbf{x} + [\mathbf{y}]_1^{i}) - f(\mathbf{x} +   [\mathbf{y}]_1^{i-1}) = f(\mathbf{x} +  \mathbf{y}) - f(\mathbf{x}).
\end{align*}
Assume now $f$ has submodularity ratio $\eta = 1$. We prove that $f$ is weakly DR-submodular by proving that it is submodular. Hence, we want to prove that for all $\mathbf{x} \in \setbf{X}$, $\forall i \neq j$ and $a_i,a_j >0$ s.t. $x_i + a_i \in \mathcal{X}_i$, $x_j + a_j \in \mathcal{X}_j$,
\begin{align}\label{sub_cond}
f(\mathbf{x} + a_i \mathbf{e}_i ) - & f(\mathbf{x})  \geq    \\&  f(\mathbf{x} + a_i \mathbf{e}_i+  a_j \mathbf{e}_j) - f(\mathbf{x} + a_j \mathbf{e}_j) \,. \nonumber
\end{align}
Consider $\mathbf{y} = a_i\mathbf{e}_i + a_j \mathbf{e}_j \in \setbf{X}$. Since $f$ has submodularity ratio $\eta = 1$, we have 
\begin{align*}
f(\mathbf{x} + a_i \mathbf{e}_i ) - f(\mathbf{x}) + f(\mathbf{x} + a_j \mathbf{e}_j ) - f(\mathbf{x}) \\ \geq    f(\mathbf{x} + a_i \mathbf{e}_i  + a_j \mathbf{e}_j) - f(\mathbf{x}) \, ,
\end{align*}
which is equivalent to the submodularity condition \eqref{sub_cond}.
\end{proof}

An example of functions with submodularity ratio $\eta >0$ is the product between an affine and a weakly DR-submodular function, as stated in the following proposition. 

\begin{proposition}
\label{prop:product_modular_weak_submodular}
Let $f,\rho : \setbf{X} \subseteq \reals{n}_+ \rightarrow \reals{}_+$ be two monotone functions, with $f$ weakly DR-submodular, and $g$ affine such that $\rho(\mathbf{x}) = \mathbf{a}^\top \mathbf{x} +b$ with $\mathbf{a}\geq \mathbf{0}$ and $b>0$. Then, provided that $\setbf{X}$ is bounded, the product $g(\mathbf{x}) := f(\mathbf{x}) \rho(\mathbf{x})$ has submodularity ratio $\eta = \inf_{i\in [n], \mathbf{x} \in \setbf{X}} \frac{b}{b + \sum_{j\neq i}a_j x_j} >0$.
\end{proposition}
\begin{proof}
Note that since $\rho$ is affine, for any $\mathbf{x}, \mathbf{y} \in \setbf{X}$ we have that $g( \mathbf{x} +\mathbf{y}) -g( \mathbf{x})   =  f( \mathbf{x} +\mathbf{y}) \rho( \mathbf{x} +\mathbf{y}) - f( \mathbf{x})\rho( \mathbf{x})=\rho( \mathbf{x} +\mathbf{y}) [ f( \mathbf{x} +\mathbf{y})  - f( \mathbf{x})] + f(\mathbf{x}) \:( \mathbf{a}^\top \mathbf{y})$. For any pair $\mathbf{x}, \mathbf{y} \in \setbf{X}$  we have: 
\begin{align*}
&\sum_{i=1}^n \big[ g(\mathbf{x} + y_i \mathbf{e}_i) - g(\mathbf{x}) \big]  \\
& = \sum_{i=1}^n \rho( \mathbf{x} +y_i \mathbf{e}_i) [ f( \mathbf{x} +y_i \mathbf{e}_i)  - f( \mathbf{x})] + f(\mathbf{x}) \:( y_i \mathbf{a}^\top  \mathbf{e}_i) \\ 
& \geq \min_{i \in [n]}\rho(\mathbf{x} +y_i \mathbf{e}_i)\sum_{i=1}^n f( \mathbf{x} +y_i \mathbf{e}_i)  - f( \mathbf{x}) + f(\mathbf{x}) \:( \mathbf{a}^\top  \mathbf{y})\\
& \geq \underbrace{\frac{\min_{i \in [n]}\rho(\mathbf{x} +y_i \mathbf{e}_i)}{\rho(\mathbf{x} + \mathbf{y})}}_{:= \eta(\mathbf{x}, \mathbf{y})} \Big( \rho( \mathbf{x} +\mathbf{y}) [ f( \mathbf{x} +\mathbf{y}) - f( \mathbf{x})]  \\[-2.5em]
& \phantom{\frac{\min_{i \in [n]}\rho(\mathbf{x} +y_i \mathbf{e}_i)}{\rho(\mathbf{x} + \mathbf{y})} \Big(\qquad \qquad } \: + f(\mathbf{x}) \:( \mathbf{a}^\top \mathbf{y}) \Big)\\
&  =\eta(\mathbf{x}, \mathbf{y}) \: [g( \mathbf{x} +\mathbf{y}) -g( \mathbf{x})] \,.
\end{align*}
The first inequality follows since $\rho$ is affine non-negative and $f$ is non-negative. The second inequality is due to $f$ being weakly DR-submodular ($f$ has submodularity ratio $\eta = 1$) and $0 < \eta(\mathbf{x}, \mathbf{y}) \leq 1$, which holds because $b > 0$ and $\mathbf{a} \geq \mathbf{0}$. Hence, it follows that $\gamma$ has submodularity ratio
\begin{equation*}
\eta := \inf_{\substack{\mathbf{x}, \mathbf{y} \in \setbf{X} :\\ \mathbf{x} + \mathbf{y} \in \setbf{X}  }}{  \eta(\mathbf{x}, \mathbf{y}) } = \inf_{i\in [n], \mathbf{y} \in \setbf{X}} \frac{b}{b + \sum_{j\neq i}a_j y_j} >0 \,. 
\end{equation*}
\end{proof}

\subsubsection{Related notion by \cite{hassani2017}}
A generalization of submodular continuous functions was also provided in \cite{hassani2017} together with provable maximization guarantees. However, it has different implications than the submodularity ratio defined above. In fact, \cite{hassani2017} considered the class of differentiable functions $f: \setbf{X} \subseteq \reals{n}_+ \rightarrow \mathbb{R}$ with parameter $\eta$ defined as
\begin{equation*}
\eta = \inf_{\mathbf{x}, \mathbf{y} \in \setbf{X},\mathbf{x} \leq \mathbf{y}  } \inf_{i \in [n]} \frac{[\nabla f(\mathbf{x}) ]_i}{[\nabla f(\mathbf{y}) ]_i} \,.
\end{equation*}
For monotone functions $\eta \in [0,1]$, and a differentiable function is DR-submodular iff $\eta = 1$ \cite{hassani2017}. Note that the parameter $\eta$ of \cite{hassani2017} generalizes the \texttt{DR} property of $f$, while our submodularity ratio $\eta$ generalizes the \texttt{weak DR} property.

\subsection{Relations with the generalized submodularity ratio of \refdef{def:playerwise-submodularity_ratio}}
In \refprop{prop:weak_submodularity_ratio} we saw that submodularity ratio $\eta =1$ is a necessary and sufficient condition for weak DR-submodularity. In contrast, a generalized submdoularity ratio (\refdef{def:playerwise-submodularity_ratio}) $\eta = 1$ is only necessary for the social function $\gamma$ to be weakly DR-submodular. This is stated in the following proposition.  For non submodular $\gamma$, no relation can be established between submodularity ratio of \refdef{def:weak_DR_sub_ratio} and generalized submodularity ratio of \refdef{def:playerwise-submodularity_ratio}.

\begin{proposition}
\label{prop:weak_DR_implies_playerwise_ratio_one}
Given a game $\set{G}=(N, \{ \set{S}_i\}_{i=1}^N,  \{ \pi_i\}_{i=1}^N , \gamma)$. If $\gamma$ is weakly DR-submodular, then  $\gamma$ has generalized submodularity ratio $\eta = 1$.
\end{proposition}
\begin{proof}
Consider any pair of outcomes $\mathbf{s}, \mathbf{s}' \in \setbf{S}$. For $i \in \{0, \ldots, N\}$, with abuse of notation we define $[\mathbf{s}']_1^i := (\mathbf{s}'_1,\ldots, \mathbf{s}'_i,\mathbf{0}, \ldots,\mathbf{0})$ with $[ \mathbf{s}']_1^0 = \mathbf{0}$. We have, 
\begin{align*}
 & \sum_{i=1}^N \gamma(\mathbf{s}_i +  \mathbf{s}'_i, \mathbf{s}_{-i}) - \gamma(\mathbf{s})\\
& \geq  \sum_{i=1}^N \gamma( \mathbf{s} +[ \mathbf{s}']_1^i) - \gamma(\mathbf{s} +[ \mathbf{s}']_1^{i-1} ) \\
& = \gamma(\mathbf{s} + \mathbf{s}') - \gamma(\mathbf{s}) \, ,
\end{align*}
where the inequality follows since $\gamma$ is weakly DR-submodular and the equality is a telescoping sum.
\end{proof}

Similarly to \refprop{prop:product_modular_weak_submodular} in the previous section, in the following proposition we show that social functions $\gamma$ defined as product of weakly DR-submodular functions and affine functions have generalized submodularity ratio $\eta >0$.

\begin{proposition}
\label{prop:product_modular_weak_submodular_playerwise}
Given a game $\set{G}=(N, \{ \set{S}_i\}_{i=1}^N,  \{ \pi_i\}_{i=1}^N , \gamma)$. Let $\gamma$ be defined as $\gamma(\mathbf{s}) := f(\mathbf{x}) \rho(\mathbf{x})$ with $f,\rho : \reals{Nd}_+ \rightarrow \reals{}_+$ be two monotone functions, with $f$ weakly DR-submodular, and $g$ affine such that $\rho(\mathbf{x}) = \mathbf{a}^\top \mathbf{x} +b$ with $\mathbf{a} = (\mathbf{a}_1, \ldots, \mathbf{a}_N)\geq \mathbf{0}$ and $b>0$. Then, $\gamma$ has generalized submodularity ratio $\eta = \inf_{i\in [N], \mathbf{s} \in \setbf{S}} \frac{b}{b + \sum_{j\neq i}\mathbf{a}_j^\top \mathbf{s}_j} >0$.
\end{proposition}

\begin{proof}
The proof is equivalent to the proof of \refprop{prop:product_modular_weak_submodular}, with the only difference that $ \mathbf{s}'_i$ belong to $\reals{d}_+$ instead of $\reals{}_+$.
\end{proof}

Note that for the game considered in the previous proposition, using  \refprop{prop:product_modular_weak_submodular} one could also affirm that $\gamma$ has submodularity ratio $\eta = \inf_{i\in [Nd], \mathbf{s} \in \setbf{S}} \frac{b}{b + \sum_{j\neq i} [\mathbf{a}]_j [\mathbf{s}]_j} >0$ which, unless $d=1$, is strictly smaller than its generalized submodularity ratio.

\subsection{Proof of \refthm{thm:weak_theorem}}
The proof is equivalent to the proof of \refthm{thm:main_theorem}, with the only difference that here we prove that $\set{G}$ is a \emph{($ \eta,  \eta$)-smooth} game in the framework of \cite{roughgarden2015}. Then, it follows that $PoA_{CCE} \leq (1 + \eta)/\eta$. 

For the smoothness proof, consider any pair of outcomes $\mathbf{s},\mathbf{s}^\star \in \setbf{S}$. We have: 
\begin{align*}
& \sum_{i=1}^N \pi_i(\mathbf{s}_i^\star, \mathbf{s}_{-i}) \geq \sum_{i=1}^N \gamma(\mathbf{s}_i^\star, \mathbf{s}_{-i}) - \gamma(0, \mathbf{s}_{-i})  \\ 
& \geq   \sum_{i=1}^N \gamma(\mathbf{s}_i^\star + \mathbf{s}_i, \mathbf{s}_{-i}) - \gamma(\mathbf{s})  \\
& =  \eta \: \gamma(\mathbf{s} + \mathbf{s}^\star ) - \eta \: \gamma(\mathbf{s}) \,. 
\end{align*}
The first inequality is due to condition ii) of \refdef{def:valid_continuous_game}. The second inequality follows since $\gamma$ is playerwise DR-submodular (applying \refprop{prop:group_DR_property} for each player $i$) and the second inequality from $\gamma$ having generalized submodularity ratio $\eta$.
\hfill \qed
 
\subsection{Analysis of the sensor coverage problem with non-submodular objective}\label{sec:extension_example}
 We analyze the sensor coverage problem with non-submodular objective defined in \refsec{sec:extension}, where $\gamma(\mathbf{x}) = \sum_{r \in [d]}w_r(\mathbf{x}) \: P(r,\mathbf{x})$ with $w_r(\mathbf{x}) = \mathbf{a}_r \: \frac{\sum_{i=1}^N[\mathbf{x}_i]_r}{N} + b_r$.  Note that by \refprop{prop:product_modular_weak_submodular_playerwise}, the function $\gamma_r(\mathbf{x}):=  w_r(\mathbf{x}) \: P(r,\mathbf{x})$ has generalized submodularity ratio $\eta >0$, hence it is not hard to show that $\gamma(\mathbf{x}) = \sum_{r \in [d]} \gamma_r(\mathbf{x})$ shares the same property. Moreover, there exist parameters $\mathbf{a}_r,b_r$ for which $\gamma$ is not submodular. Interestingly, $\gamma$ is convave in each $\set{X}_i$. In fact, $\gamma_r$'s are concave in each $\set{X}_i$ since $P(r,\mathbf{x})$'s are concave in each $\set{X}_i$ and $w_r$'s are positive affine functions. Moreover, $\gamma$ is playerwise DR-submodular since the $(d\times d)$ blocks on the diagonal of its Hessian are diagonal (and their entries are non-positive, by concavity of $\gamma$ in each $\set{X}_i$). 
 
To maximize $\gamma$, as outlined in \refsec{sec:distributed_maximization}, we can set up a game $\set{G}=(N, \{ \set{S}_i\}_{i=1}^N,  \{ \pi_i\}_{i=1}^N , \gamma)$ where for each player $i$, $\set{S}_i = \set{X}_i$, and $\pi_i(\mathbf{s}) = \gamma(\mathbf{s}) - \gamma(\mathbf{0}, \mathbf{s}_{-i})$ for every outcome $\mathbf{s} \in \setbf{S} = \setbf{X}$. Hence, condition ii) of \refdef{def:valid_continuous_game} is satisfied with equality. 
Following the proof of \refthm{thm:weak_theorem}, we have that: 
\begin{equation*}
\sum_{i=1}^N \pi_i(\mathbf{s}_i^\star, \mathbf{s}_{-i}) \geq \eta \: \gamma(\mathbf{s} + \mathbf{s}^\star ) - \eta \: \gamma(\mathbf{s})
\end{equation*}
In order to bound $PoA_{CCE}$, the last proof steps of \refsec{sec:proof_of_thm1} still ought to be used. Such steps rely on condition iii), which in \refsec{sec:distributed_maximization} was proved using submodularity of $\gamma$. Although $\gamma$ is not submodular, we prove a weaker version of condition iii) as follows.  By definition of $\gamma_r$ and for every outcome $\mathbf{x}$ we have $\sum_{i=1}^N \gamma_r(\mathbf{s}) - \gamma_r(\mathbf{0},\mathbf{s}_{-i}) = \sum_{i=1}^N  w_r(\mathbf{x})[P(r,\mathbf{s}) - P(r, (\mathbf{0},\mathbf{s}_{-i})) ] + [w_r(\mathbf{s}_i, \mathbf{0})  -w_r(\mathbf{0})] P(r, (\mathbf{0},\mathbf{s}_{-i})) \leq   w_r(\mathbf{x}) P(r,\mathbf{s}) + P(r,\mathbf{s}) \sum_{i=1}^N [w_r(\mathbf{s}_i, \mathbf{0})  -w_r(\mathbf{0})] = (1 + \frac{ w_r(\mathbf{x}) - w_r(\mathbf{0}) }{w_r(\mathbf{x})}) \gamma_r (\mathbf{x}) \leq 2 \gamma_r (\mathbf{x})$. The equalities are due to $w_r$ being affine, the first inequality is due to $P(r,\cdot)$ being submodular and monotone, and the last inequality holds since $w_r$ is positive and monotone. Hence, from the inequalities above we have $\sum_{i=1}^N \pi_i(\mathbf{x}) = \sum_{i=1}^N \gamma(\mathbf{x}) - \gamma(\mathbf{0}, \mathbf{x}_{-i}) \leq 2 \gamma(\mathbf{x})$. Note that a tighter condition can also be derived depending on the functions $w_r$'s, using $(1 + \max_{\mathbf{x} \in \setbf{X},r\in [d]}\frac{ w_r(\mathbf{x}) - w_r(\mathbf{0}) }{w_r(\mathbf{x})})$ in place of 2. We will now use such condition in the same manner condition iii) was used in \refsec{sec:proof_of_thm1}.  Let $\mathbf{s}^\star = \argmax_{\mathbf{s}\in \setbf{S}}\gamma(\mathbf{s})$. Then, for any CCE $\sigma$ of $\set{G}$ we have
\begin{align*}
\mathbb{E}_{\mathbf{s} \sim \sigma} [ \gamma(\mathbf{s})] & \geq \frac{1}{2}\sum_{i=1}^N \mathbb{E}_{\mathbf{s} \sim \sigma} [ \pi_i(\mathbf{s})]  \geq \frac{1}{2}\sum_{i=1}^N \mathbb{E}_{\mathbf{s} \sim \sigma} [ \pi_i(\mathbf{s}_i^\star, \mathbf{s}_{-i})] \\
& \geq \frac{\eta}{2} \gamma( \mathbf{s}^\star )  - \frac{\eta}{2} \: \mathbb{E}_{\mathbf{s} \sim \sigma} [ \gamma(\mathbf{s})] \, .
\end{align*}
Hence, $PoA_{CCE} \leq (1 + 0.5\eta)/0.5\eta$.

\end{document}